\documentclass[twoside,final,11pt]{entics} 
\usepackage{enticsmacro}

\usepackage{amsmath,amssymb}
\usepackage{accents}
\usepackage{arydshln}
\usepackage{enumitem}					
	\setlist[1]{leftmargin=\parindent}	
\usepackage{listings}
\usepackage{mathtools}
\usepackage{proof}
\usepackage{stmaryrd}
\usepackage{thmtools}
\usepackage{tikz}
	\usetikzlibrary{arrows.meta}
	\usetikzlibrary{calc}
\usepackage{xcolor}

\makeatletter

\newcommand\black{\color{black}}

\colorlet{0}{black}
\colorlet{1}{red!80!black}
\colorlet{2}{blue!80!black}
\colorlet{3}{orange!90!black}
\colorlet{4}{green!50!black}
\colorlet{5}{violet!80!black}
\colorlet{6}{teal!90!white}
\colorlet{7}{brown!90!black}
\colorlet{8}{pink!70!blue}
\colorlet{9}{lime!80!black}

\colorlet{term}{1}
\colorlet{type}{2}
\colorlet{ctxt}{3}
\colorlet{simp}{4}
\colorlet{patn}{5}
\colorlet{ptpe}{6}

\newcommand\colorterm{\color{term}}
\newcommand\colortype{\color{type}}
\newcommand\colorctxt{\color{ctxt}}
\newcommand\colorsimp{\color{simp}}
\newcommand\colorpatn{\color{patn}}
\newcommand\colorptpe{\color{ptpe}}

\newcommand\colored{%
  \let\color@term\colorterm%
  \let\color@type\colortype%
  \let\color@ctxt\colorctxt%
  \let\color@simp\colorsimp%
  \let\color@patn\colorpatn%
  \let\color@ptpe\colorptpe%
}
\newcommand\uncolored{%
  \let\color@term\relax%
  \let\color@type\relax%
  \let\color@ctxt\relax%
  \let\color@simp\relax%
  \let\color@patn\relax%
  \let\color@ptpe\relax%
  \black
}
\uncolored

\newcommand\termcolor[1]{{\color{term}#1}}

\newcommand\vc[1]{\vcenter{\hbox{$#1$}}}

\newcommand\rvecup[1]{\accentset{\rightharpoonup}{#1}}
\newcommand\lvecup[1]{\accentset{\leftharpoonup}{#1}}

\newcommand\define[1]{\emph{#1}}
\newcommand\defeq{\stackrel{\scriptscriptstyle\Delta}=}

\newcommand\smallbin[1]{\mathchoice
      {\mathbin{\raise.2ex \hbox{$\scriptstyle      #1$}}}%
      {\mathbin{\raise.2ex \hbox{$\scriptstyle      #1$}}}%
      {\mathbin{\raise.12ex\hbox{$\scriptscriptstyle#1$}}}%
      {\mathbin{           \hbox{$\scriptscriptstyle#1$}}}}%

\newcommand\Con{\wedge}
\newcommand\Imp{\rightarrow}
\newcommand\Tim{\times}
\newcommand\Ten{\otimes}

\newcommand\con{\kern1pt{\smallbin\Con}\kern1pt}
\newcommand\imp{\kern1pt{\smallbin\Imp}\kern1pt}
\newcommand\tim{\kern1pt{\smallbin\Tim}\kern1pt}
\newcommand\ten{\kern1pt{\smallbin\Ten}\kern1pt}
\newcommand\arrimp{\kern1pt{\smallbin\rightsquigarrow}\kern1pt}

\newlength{\ilength}
\newcommand\From{\smallbin\Leftarrow}

\newcommand\fv[1]{\mathsf{fv}(#1)}
\newcommand\bv[1]{\mathsf{bv}(#1)}

\newcommand\cbn{{\mathsf{cbn}}}
\newcommand\cbv{{\mathsf{cbv}}}
\newcommand\cbpv{{\mathsf{cbpv}}}

\newcommand\uterm[1]{\underline{\term{#1}}}
\newcommand\utrm[1]{\underline{\trm{#1}}}

\newcommand\seq@discard[1]{\seq@start}
\newcommand\seq@comma{\@ifnextchar,{{\black,\dots,}~\seq@discard}{{\black,}~\seq@start}}

\newcommand\vec@space{\kern0.4pt}

\newcommand\type[1]{{\colored\typ{#1}}}
\newcommand\typ[1]{%
  \vphantom)%
  \let\seq@start\type@start%
  \seq@start#1\@end%
}
\newcommand\type@start{\color@type\let\type@loop\type@next\type@loop}
\newcommand\type@next[1]{%
  \ifx#1\@end\let\type@loop\type@end\else%
  \ifx#1`\let\type@loop\type@escape\else%
  \ifx#1_\let\type@loop\type@sub\else%
  \ifx#1^\let\type@loop\type@sup\else%
  \ifx#1!\let\type@loop\type@rvec\else%
  \ifx#1?\let\type@loop\type@lvec\else%
  \ifx#1|\let\type@loop\type@bar\else%
  \ifx#1,\let\type@loop\seq@comma\else%
  \ifx#1>{\kern1pt\smallbin\Rightarrow\kern1pt}\else%
  \ifx#1.\kern1pt{\cdot}\kern1pt\else%
  \ifx#1A\!_A\else%
  \typevar{#1}%
  \fi\fi\fi\fi\fi\fi%
  \fi\fi\fi\fi\fi%
  \type@loop%
}
\newcommand\type@sub[1]{_{#1}\type@start}
\newcommand\type@sup[1]{^{{\type@sup@color #1}}\type@start}
\newcommand\type@rvec[1]{\vec@space\rvecup{\typevar{#1}\vec@space}\type@start}
\newcommand\type@lvec[1]{\vec@space\lvecup{\typevar{#1}\vec@space}\type@start}
\newcommand\type@bar{\@ifnextchar-{~{\black\vdash}~\seq@discard}{|\type@start}}
\newcommand\type@escape[1]{#1\type@start}
\newcommand\type@end{\uncolored}

\newcommand\typevar[1]{\typevar@start#1\relax}
\newcommand\typevar@start{\let\typevar@loop\typevar@next\typevar@loop}
\newcommand\typevar@next[1]{%
  \ifx#1\relax\let\typevar@loop\typevar@end\else%
  \ifx#1k\kappa\else%
  \ifx#1r\rho\else%
  \ifx#1s\sigma\else%
  \ifx#1t\tau\else%
  \ifx#1u\upsilon\else%
  \ifx#1w\omega\else%
  #1%
  \fi\fi\fi\fi\fi\fi\fi%
  \typevar@loop%
}
\newcommand\typevar@end{}
  
\newcounter{@parens}

\newcommand\term[1]{{\colored\trm{#1}}}
\newcommand\trm[1]{%
  \setcounter{@parens}{0}%
  \vphantom(%
  \let\seq@start\term@start%
  \seq@start#1\@end%
}
\newcommand\term@start{\color@term\let\term@loop=\term@next\term@loop}
\newcommand\term@next[1]{%
  \ifx#1\@end\let\term@loop\term@end\else%
  \ifx#1`\let\term@loop\term@escape\else%
  \ifx#1"\let\term@loop\term@color\else%
  \ifx#1:\let\term@loop\term@colon\else%
  \ifx#1|\let\term@loop\term@bar\else%
  \ifx#1((\stepcounter{@parens}\else%
  \ifx#1))\addtocounter{@parens}{-1}\else%
  \ifx#1\{\{\stepcounter{@parens}\else%
  \ifx#1\}\}\addtocounter{@parens}{-1}\else%
  \ifx#1,\ifnum\value{@parens}=0\let\term@loop\seq@comma\else,\fi\else%
  \ifx#1_\let\term@loop\term@sub\else%
  \ifx#1^\let\term@loop\term@sup\else%
  \ifx#1?\let\term@loop\term@lvec\else%
  \ifx#1!\let\term@loop\term@rvec\else%
  \ifx#1*\star\else%
  \ifx#1l\lambda\else%
  \ifx#1<\langle\else%
  \ifx#1>\rangle\else%
  \ifx#1..\,\else%
  \ifx#1;\,{;}\,\else%
  \ifx#1=\kern1pt{\smallbin=}\kern1pt\else%
  \ifx#1G{{\black\Gamma}}\else%
  \ifx#1D{{\black\Delta}}\else%
  \ifx#1\sim{\,\black\sim\,}\else%
  #1%
  \fi\fi\fi\fi\fi%
  \fi\fi\fi\fi\fi%
  \fi\fi\fi\fi%
  \fi\fi\fi\fi\fi%
  \fi\fi\fi\fi\fi%
  \term@loop%
}
\newcommand\term@end{\color{black}\uncolored}
\newcommand\term@discard[1]{\term@start}
\newcommand\term@escape[1]{#1\term@start}
\newcommand\term@color[1]{\color{#1}\term@start}
\newcommand\term@sub[1]{_{#1}\term@start}
\newcommand\term@sup[1]{^{#1}\term@start}
\newcommand\term@rvec[1]{\rvecup{#1}\term@start}
\newcommand\term@lvec[1]{\lvecup{#1}\term@start}
\newcommand\term@bar{\@ifnextchar-{~{\black\vdash}~\term@discard}{|\term@start}}
\newcommand\term@colon{\@ifnextchar={\term@assign}{\term@type}}
\newcommand\term@assign{\mathbin{{:}{=}}\term@discard}
\newcommand\term@type{\black\colon\type@start}

\newcommand\cat[1]{
  \cat@start#1\@end%
}
\newcommand\cat@start{\let\cat@loop\cat@next\cat@loop}
\newcommand\cat@next[1]{%
  \ifx#1\@end\let\cat@loop\cat@end\else%
  \ifx#1_\let\cat@loop\cat@sub\else%
  \ifx#1^\let\cat@loop\cat@sup\else%
  \ifx#1?\let\cat@loop\cat@lvec\else%
  \ifx#1-\let\cat@loop\cat@arrows\else%
  \ifx#1*\tim\else%
  \ifx#1A\!_A\else%
  \typevar{#1}%
  \fi\fi\fi\fi\fi
  \fi\fi%
  \cat@loop%
}
\newcommand\cat@sub[1]{_{#1}\cat@start}
\newcommand\cat@sup[1]{^{#1}\cat@start}
\newcommand\cat@rvec[1]{\rvecup{\typevar{#1}}\cat@start}
\newcommand\cat@lvec[1]{\lvecup{\typevar{#1}}\cat@start}
\newcommand\cat@arrows{%
  \@ifnextchar-{\longrightarrow\cat@discard@two}{%
  \@ifnextchar>{\imp\cat@discard}{-\cat@start}}}
\newcommand\cat@end{}
\newcommand\cat@discard[1]{\cat@start}
\newcommand\cat@discard@two[2]{\cat@start}

\newcommand\rr[1]{%
  \ifx#1x\mathsf{var}\else%
  \ifx#1l\mathsf{abs}\else%
  \ifx#1a\mathsf{arg}\else%
  \ifx#1f\mathsf{fun}\else%
  \fi\fi\fi\fi%
}

\newcommand\get{\mathsf{get}~}
\newcommand\set{\mathsf{set}~}
\newcommand\print{\mathsf{print}}
\newcommand\rand{\mathsf{rand}}
\newcommand\rd{\mathsf{read}}

\newcommand\wrt{\mathsf{write}~}

\newcommand\ifthen{\mathsf{if}}

\newcommand\Z{\mathbb Z}
\newcommand\plus{+}

\newlength\diamondsize
\newcommand\@diamond[2]{\tikz\draw[#2](-#1,0)--(0,#1)--(#1,0pt)--(0,-#1)--cycle;}
\newcommand\@ldiamond[2]{\tikz\draw[#2](-#1,0)--(0,#1)--(0,-#1)--cycle;}
\newcommand\@rdiamond[2]{\tikz\draw[#2](0,#1)--(#1,0pt)--(0,-#1)--cycle;}

\newcommand\filldiamond{\mathchoice%
  {\@diamond{2.0pt}{fill}}%
  {\@diamond{2.0pt}{fill}}%
  {\@diamond{1.5pt}{fill}}%
  {\@diamond{1.2pt}{fill}}}
\newcommand\leftfilltriangle{\mathchoice%
  {\@ldiamond{2.0pt}{fill}}%
  {\@ldiamond{2.0pt}{fill}}%
  {\@ldiamond{1.5pt}{fill}}%
  {\@ldiamond{1.2pt}{fill}}}
\newcommand\rightfilltriangle{\mathchoice%
  {\@rdiamond{2.0pt}{fill}}%
  {\@rdiamond{2.0pt}{fill}}%
  {\@rdiamond{1.5pt}{fill}}%
  {\@rdiamond{1.2pt}{fill}}}

\newcommand\opendiamond{\mathchoice%
  {\@diamond{2.0pt}{}}%
  {\@diamond{2.0pt}{}}%
  {\@diamond{1.5pt}{}}%
  {\@diamond{1.2pt}{}}}
\newcommand\leftopentriangle{\mathchoice%
  {\@ldiamond{2.0pt}{}}%
  {\@ldiamond{2.0pt}{}}%
  {\@ldiamond{1.5pt}{}}%
  {\@ldiamond{1.2pt}{}}}
\newcommand\rightopentriangle{\mathchoice%
  {\@rdiamond{2.0pt}{}}%
  {\@rdiamond{2.0pt}{}}%
  {\@rdiamond{1.5pt}{}}%
  {\@rdiamond{1.2pt}{}}}

\newcommand\MXL{\leftfilltriangle}
\newcommand\MXR{\rightfilltriangle}
\newcommand\MOL{\leftopentriangle}
\newcommand\MOR{\rightopentriangle}

\newcommand\MX{\filldiamond}
\newcommand\MO{\opendiamond}

\newcommand\llb\llbracket
\newcommand\rrb\rrbracket

\newcommand\llp{\llparenthesis\kern1pt}
\newcommand\rrp{\kern1pt\rrparenthesis}

\renewcommand\e{\varepsilon}

\newcommand\loc[1]{\mathsf{loc}(\term{#1})}

\newcommand\rnd{\mathsf{rnd}} 
\newcommand\nd {\mathsf{nd}}  
\newcommand\inp{\mathsf{in}}  
\newcommand\out{\mathsf{out}} 

\newcommand\FMC{\textnormal{FMC}}

\tikzstyle{rwhead}=[>/.tip={Triangle[open,length=2.5pt,width=4.5pt]},|/.tip={Rectangle[length=.5pt,width=4.5pt]}]

\tikzstyle{rw} =[line width=.5pt,rwhead,->]
\tikzstyle{rws}=[line width=.5pt,rwhead,->.>]
\tikzstyle{rwn}=[line width=.5pt,rwhead,->.>|]
\tikzstyle{rwp}=[line width=.5pt,rwhead,->,double]
\tikzstyle{rwps}=[line width=.5pt,rwhead,->.>,double]

\newcommand\rw  {\mathrel{\tikz\draw[rw]  (0,0)--(10pt,0pt);}}
\newcommand\rws {\mathrel{\tikz\draw[rws] (0,0)--(10pt,0pt);}}

\newcommand\rwp {\mathrel{\tikz\draw[rwp] (0,0)--(10pt,0pt);}}
\newcommand\rwps{\mathrel{\tikz\draw[rwps](0,0)--(10pt,0pt);}}

\newcommand\rwls [1][\relax]{\mathrel{_{#1}\tikz\draw[rws] (10pt,0pt)--(0,0);}}

\newcommand\rwlp [1][\relax]{\mathrel{_{#1}\tikz\draw[rwp] (10pt,0pt)--(0,0);}}
\newcommand\rwlps[1][\relax]{\mathrel{_{#1}\tikz\draw[rwps](10pt,0pt)--(0,0);}}

\newcommand\val{{v}}

\newcommand\ttret{\texttt{return}} 
\newcommand\ttlet{\texttt{\upshape let}}

\newcommand\ttto{\texttt{to}}
\newcommand\ttforce{\texttt{force}}
\newcommand\ttthunk{\texttt{thunk}}

\newcommand\ttin{\texttt{\upshape in}}

\newcommand\ac{\kern1pt{\smallbin\ggg}\kern1pt} 
\newcommand\arr{\mathsf{arr}}
\newcommand\first{\mathsf{first}}

\newcommand\TR[1]{\!\scriptstyle{\tr{#1}\vphantom{pb}}}%
\newcommand\tr[1]{%
  \ifx#1*\star\else%
  \ifx#1x\mathsf{var}\else%
  \ifx#1l\mathsf{abs}\else%
  \ifx#1a\mathsf{app}\else%
  \ifx#1c\mathsf{cut}\else%
  \ifx#1f\mathsf{con}\else%
  \ifx#1<\mathsf{lcut}\else%
  \ifx#1>\mathsf{rcut}\else%
  \fi\fi\fi\fi\fi\fi\fi\fi%
}

\newcommand\qr[1]{%
  \ifx#1*\star\else%
  \ifx#1x\mathsf{AX}\else%
  \ifx#1l\mathsf{ABS}\else%
  \ifx#1a\mathsf{APP}\else%
  \ifx#1c\mathsf{CUT}\else%
  \ifx#1f\mathsf{con}\else%
  \ifx#1<\mathsf{lcut}\else%
  \ifx#1>\mathsf{rcut}\else%
  \fi\fi\fi\fi\fi\fi\fi\fi%
}

\newcommand\sr[1]{\vphantom(%
  \ifx#1l{\imp}\!\mathsf R\else%
  \ifx#1a{\imp}\!\mathsf L\else%
  \ifx#1x\mathsf{ax}\else%
  \ifx#1c\mathsf{cut}\else%
  \ifx#1V{\forall}\mathsf R\else%
  \ifx#1A{\forall}\mathsf L\else%
  \ifx#1e\textsf{eta}\else%
  \ifx#1z\textsf{cls}%
  \fi\fi\fi\fi\fi\fi\fi\fi%
}

\newcommand\eval{\kern1pt{\Downarrow}\kern1pt}

\newcommand\machine{\@ifstar\@textmachine\@mathmachine}

\newcommand\@textmachine[4]{(#1,\term{#2})\eval(#3,\term{#4})}

\newcommand\@mathmachine[4]{%
\begin{array}{@{(~}l@{~,~}r@{~)}}%
#1 & \term {#2}%
\\\hline\hline\rule[-5pt]{0pt}{15pt}%
#3 & \if#4*\term*\;\else\term{#4}\fi%
\end{array}%
}

\newcommand\step[4]{%
\begin{array}{@{(~}l@{~,~}r@{~)}}%
#1 & \term {#2}%
\\\hline%
#3 & \if#4*\term*\;\else\term{#4}\fi%
\end{array}%
}

\newcommand\run[1]{\textsc{run}(\type{#1})}

\newcommand\diagdots[2][1.5,3]{
  \node[fill=black,circle,inner sep=0pt,minimum size=1.5pt] at ($(#2) - (#1)$) {};
  \node[fill=black,circle,inner sep=0pt,minimum size=1.5pt] at (#2) {};
  \node[fill=black,circle,inner sep=0pt,minimum size=1.5pt] at ($(#2) + (#1)$) {};
}
\tikzstyle{termbox}=[draw=term,fill=term!10,rounded corners,minimum size=20pt]
\tikzstyle{tallbox}=[draw=term,fill=term!10,rounded corners,minimum width=20pt,minimum height=40pt]
\tikzstyle{termpic}=[x=1pt,y=1pt,inner sep=0pt,outer sep=0pt,thick]

\newcommand\wires[5]{
  \node[anchor=east] at (#2,#3) {$\type{#1}$};
  \node[anchor=west] at (#4,#3) {$\type{#5}$};
  \draw ($(#2,#3) + ( 2, 6)$) -- ($(#4,#3) + (-2, 6)$);
  \draw ($(#2,#3) + ( 8,-6)$) -- ($(#4,#3) + (-8,-6)$);
  \diagdots[-1.5,3]{$(#2,#3) + (10,0)$}
  \diagdots[ 1.5,3]{$(#4,#3) - (10,0)$}
}
\newcommand\wiresright[4]{
  \node[anchor=west] at (#3,#2) {$\type{#4}$};
  \draw ($(#1,#2) + ( 0, 6)$) -- ($(#3,#2) + (-2, 6)$);
  \draw ($(#1,#2) + ( 0,-6)$) -- ($(#3,#2) + (-8,-6)$);
  \diagdots[ 1.5,3]{$(#3,#2) - (10,0)$}
}
\newcommand\wiresleft[4]{
  \node[anchor=east] at (#2,#3) {$\type{#1}$};
  \draw ($(#2,#3) + ( 2, 6)$) -- ($(#4,#3) + ( 0, 6)$);
  \draw ($(#2,#3) + ( 8,-6)$) -- ($(#4,#3) + ( 0,-6)$);
  \diagdots[-1.5,3]{$(#2,#3) + (10,0)$}
}

\makeatother

\def\conf{MFPS 2022} 	
\volume{1}			
\def\volu{1}			
\def\papno{6}			
\def\mydoi{{\normalshape \href{https://doi.org/10.46298/entics.10513}{10.46298/entics.10513}}}
\def\copynames{W. Heijltjes} 

\def\lastname{Heijltjes}
\def\runtitle{Functional Machine Calculus}

\begin{document}

\begin{frontmatter}
	\title{The Functional Machine Calculus}
	\author{Willem Heijltjes\thanksref{ALL}}
	\address{Department of Computer Science\\University of Bath\\Bath, United Kingdom}
	\thanks[ALL]{Email: \href{wbh22@bath.ac.uk}{wbh22@bath.ac.uk}}
	\begin{abstract}
		This paper presents the Functional Machine Calculus (FMC) as a simple 
		model of higher-order computation with ``reader/writer'' effects: higher-order
	    mutable store, input/output, and probabilistic and non-deterministic
	    computation.	    
	    
	    The FMC derives from the lambda-calculus by taking the standard operational
	    perspective of a call--by--name stack machine as primary, and introducing two
	    natural generalizations. One, ``locations'', introduces multiple stacks, which
	    each may represent an effect and so enable effect operators to be encoded
		into the abstraction and application constructs of the calculus. The second,
		``sequencing'', is known from kappa-calculus and concatenative programming 
		languages, and introduces the imperative notions of ``skip'' and ``sequence''.
		This enables the encoding of reduction strategies, including call--by--value
		lambda-calculus and monadic constructs.
		
		The encoding of effects into generalized abstraction and application means that
		standard results from the lambda-calculus may carry over to effects. The main
		result is confluence, which is possible because encoded effects reduce 
		algebraically rather than operationally. Reduction generates the familiar 
		algebraic laws for state, and unlike in the monadic setting, reader/writer 
		effects combine seamlessly. A system of simple types confers 
		termination of the machine.
	\end{abstract}
	\begin{keyword}
		lambda-calculus, 
		computational effects, 
		confluence,
		concatenative programming
	\end{keyword}
\end{frontmatter}




\section{Introduction}

Higher-order programming and computational effects are ubiquitous in modern programs. Understanding them, and in particular their potent combination, is therefore an important challenge to computer science. Higher-order functional programming enjoys an elegant foundational theory in the $\lambda$-calculus, where $\beta$-reduction gives rise not only to operational semantics---by imposing an evaluation strategy---but also to an equational theory which may be regarded as definitive for higher-order functions. For computational effects, however, there are many approaches and, as yet, no single definitive theory. Such a theory would ideally include a convenient syntax, expressing a natural and convincing semantics, and supporting reasoning tools and methods such as type-systems and compile-time optimizations, while remaining amenable to refinement, extension and variation.

The rich history of approaches to the problem of computational effects
in a higher-order setting includes Landin's pioneering work~\cite{Landin-1965}, which
cemented the central position of $\lambda$-calculus, and highlighted
the difficulty of reconciling the drive for a complete theory with the
practice of programming: Landin focussed on a call-by-value strategy
and used \emph{thunks} to delay evaluation where necessary. A more
modular and flexible account was provided by Moggi's use of \emph{monads}~\cite{Moggi-1991},
which has influenced not only theoretical work but also the design of
the Haskell programming language. However, a fundamental problem with
monads is that they don't compose, and in practice multiple effects
are combined by building a stack of monad transformers, which can
become unwieldy for the programmer. Many alternatives and refinements
have been proposed, including 
  uniqueness types~\cite{Smetsers-Barendsen-vanEekelen-Plasmeijer-1993}, 
  continuations~\cite{Filinski-1996},
  encodings in (intuitionistic) linear logic~\cite{Benton-Wadler-1996,Maraist-Odersky-Turner-Wadler-1999}
  and in process calculi~\cite{Milner-1992,Hirschkoff-Prebet-Sangiorgi-2020}, 
  premonoidal categories~\cite{Power-Robinson-1997,Power-Thielecke-1999,Levy-Power-Thielecke-2003},
  Call--By--Push--Value~\cite{Levy-2003} 
  and its variants in linear logic~\cite{Egger-Mogelberg-Simpson-2014,Ehrhard-Guerrieri-2016},
  Arrows~\cite{Hughes-2000}, 
  algebraic effects~\cite{Plotkin-Power-2002-FOSSACS,Ahman-Staton-2013-ENTCS},
  and effect handlers~\cite{Plotkin-Pretnar-2009}.

This paper offers a new solution to this challenge: the \emph{Functional Machine Calculus} (FMC). It includes the effects of state, input/output, and probabilistic and non-deterministic computation---here referred to collectively as \emph{reader/writer effects}. The FMC consists of two independent generalizations of the $\lambda$-calculus, \emph{locations} and \emph{sequencing}, that individually give two fragments, the \emph{poly-$\lambda$-calculus} and the \emph{sequential $\lambda$-calculus}. It enjoys a clean equational theory supported by a confluent reduction semantics, and expresses both effects and higher-order features in the same terms, retaining the simplicity of the $\lambda$-calculus while being powerful enough to capture the reality of higher-order programming with multiple effects. We provide operational semantics in terms of an abstract machine, and a type system that confers termination of the machine. The remainder of this section will introduce both generalizations. Throughout the paper, proofs are omitted when they are straightforward.


\subsection{Locations}

The main objective of this work has been to preserve confluence, following the recent presentation of a confluent probabilistic $\lambda$-calculus by Dal Lago, Guerrieri, and Heijltjes~\cite{DalLago-Guerrieri-Heijltjes-2020}. This is perhaps surprising, as $\lambda$-calculi with effects are known to be non-confluent. The apparent contradiction disappears by disentangling the \emph{operational} and the \emph{algebraic} aspects of evaluation. In $\lambda$-calculus, $\beta$-reduction is algebraic (or more precisely, $\beta$-\emph{equivalence} is), while stack machines such as Krivine's~\cite{Krivine-2007} give an operational semantics. For effects, looking up a global variable or generating a random value is operational, while effect operators may interact algebraically via the laws of Plotkin and Power~\cite{Plotkin-Power-2002-FOSSACS}.

\begin{center}
\begin{tabular}{@{}l@{\qquad}l@{\qquad}l@{}}
				            & \textbf{global, operational} 	& \textbf{local, algebraic}
\\ \hline
\textbf{$\lambda$-calculus} & stack machines        		& $\beta$-reduction
\\       
\textbf{effects} 			& update, lookup, read, write, random & algebraic effect equations
\end{tabular}
\end{center}

Our starting point is the observation that for both the $\lambda$-calculus and reader/writer effects, the operational side can be given by push and pop actions on global stacks or streams. In a simple stack machine for the $\lambda$-calculus, \emph{application} $\term{M\,N}$ pushes its argument $\term N$ to the stack and continues as $\term M$, and \emph{abstraction} $\term{lx.M}$ pops a term $\term N$ from the stack and binds it to $\term x$, to continue as $\term{\{N/x\}M}$ (the capture-avoiding substitution of $\term N$ for $\term x$ in $\term M$). Then, the following effects are also modelled via stacks or streams:
\begin{itemize}
	\item
reading from \emph{input} is a pop from an input stream;
	\item 
writing to \emph{output} is a push to an output stream;
	\item
a memory cell $\term c$ is modelled by a stack of depth one, where
\begin{itemize}
	\item \emph{update} $\term{c := N}$ pops from $\term c$, discarding the value, then pushes the new value $\term N$;
	\item \emph{lookup} $\term{`!c}$ pops the value $\term N$ from $\term c$, pushes $\term N$ to reinstate $\term c$, and then returns $\term N$;
\end{itemize}
	\item
\emph{probabilistic} and \emph{non-deterministic} generators can be modelled as separate input streams.
\end{itemize}
This idea is captured in the \emph{poly-$\lambda$-calculus}: we introduce a set of \emph{locations} to represent independent stacks or streams on the machine, and parameterize abstraction and application in this set to act as pop and push actions on the corresponding stack. Effect operators are then encoded in these constructs according to the above scheme. Beta-reduction, generalized to multiple locations, remains confluent, and for encoded effect operators it gives rise to the expected algebraic laws~\cite{Plotkin-Power-2002-FOSSACS}. However, this encoding of effects forces their call--by--name semantics, while programming with effects requires control over when they are called. This is the purpose of the second generalization, \emph{sequencing}.
 

\subsection{Sequencing}

The literature offers several ways to control reduction behaviour in higher-order languages, including continuation encodings between $\cbv$ and $\cbn$~\cite{Plotkin-1975}, and call--by--push--value ($\cbpv$)~\cite{Levy-2003} which encodes both. To complement \emph{locations}, our approach takes the stack machine as primary. Viewing the $\lambda$-calculus as a language of machine instruction sequences, the \emph{sequencing} generalization extends it with composition and the empty sequence, analogous to imperative ``sequence'' and ``skip''; not by introducing these as primitives, but again by generalizing the calculus in a subtle way so that they arise naturally.
Such designs have arisen several times before: in the first-order $\kappa$-calculus of Hasegawa~\cite{Hasegawa-1995}, generalized to higher-order in the context of premonoidal categories~\cite{Power-Thielecke-1999}; as the $\Lambda_s$-calculus in an analysis of compilers~\cite{Douence-Fradet-1998}; and in higher-order stack programming languages, also called \emph{concatenative} languages~\cite{Herzberg-Reichert-2009}, such as \emph{Joy}~\cite{vonThun-2001}, \emph{$\lambda$-FORTH}~\cite{Lynas-Stoddart-2006}, \emph{Cat}~\cite{Diggins-2008}, and closest to our design, \emph{Factor}~\cite{Pestov-Ehrenberg-Groff-2010}.

Douence and Fradet demonstrate how their $\Lambda_s$ encodes Plotkin's $\cbv$ $\lambda$-calculus~\cite{Plotkin-1975} as well as Moggi's monadic constructs~\cite{Moggi-1989-LICS,Moggi-1991}, illustrating how this design gives control over reduction. We will recall these encodings in Sections~\ref{sec:cbnencodings} and~\ref{sec:cbvencodings}, and demonstrate the encoding of $\cbpv$ and Arrows~\cite{Hughes-2000}.


\subsection{The Functional Machine Calculus}
The FMC combines both generalisations, \emph{locations} and \emph{sequencing}, in a simple model of higher-order computation with multiple effects. Its design and solid foundations in semantics mean that several important properties of the $\lambda$-calculus are preserved. The aim of this paper is introductory: it presents the syntax, operational semantics, and fundamental ideas and results with an emphasis on explanation and examples. In a forthcoming  paper we will deepen these results with a strong normalization theorem, a domain-theoretic semantics of the untyped calculus, and semantics of the typed calculus in premonoidal and in Cartesian closed categories. The main results for the FMC as presented here are the following:

\begin{description}
	\item[Confluence]
Beta-reduction is confluent, with evaluation behaviour expressed in syntax.
	\item[Algebraic effects]
The algebraic laws for reader/writer effects arise from reduction.
	\item[Compositionality]
Reader/writer effects combine seamlessly, due to the use of independent locations.
	\item[Types]
Simple types cover (higher-order) effect operations and confer termination of the machine.
\end{description}


\section{The poly-lambda-calculus}

We introduce a set of \emph{locations} $A$, ranged over by $\term{a,b,c},\dots$, to indicate the different stacks or streams for each effect. Abstraction and application are parameterized in $A$ to give the corresponding pop and push actions. We write application $\term{M\,N}$ as $\term{[N].M}$ to emphasize the operational reading (push $\term N$ and continue as $\term M$), to easily attach a location $\term a$, and to give unique parsing---cf.\ De Bruijn~\cite{DeBruijn-1993}. Abstraction $\term{lx.N}$ is written $\term{<x>.N}$ to emphasize the duality with application.

\begin{definition}
The \emph{poly-$\lambda$-calculus} is given by the grammar 
\[
	\term{M,N} ~\Coloneqq~ \term{x} ~\mid~ \term{[N]a.M} ~\mid~ \term{a<x>.M}
\]
with from left to right a \emph{variable}, an \emph{application} or \emph{push action} on location $\term a$ with function $\term M$ and argument $\term N$, and an \emph{abstraction} or \emph{pop action} on location $\term a$ that binds $\term x$ in $\term M$. Terms are considered modulo $\alpha$-equivalence. The regular $\lambda$-calculus embeds via a dedicated \emph{main} location $\term l\in A$, omitted from terms for brevity; so we may write $\term{lx.M}$ or $\term{<x>.M}$ for $\term{l<x>.M}$, and $\term{M\,N}$ or $\term{[N].M}$ for $\term{[N]l.M}$.

The \emph{poly-stack machine} is given by the following data. A \emph{stack} of terms $S$ is written with the top element to the right; we define them inductively below left, but they should better be considered as coinductive, to include streams. A \emph{memory} $S_A$ is a family of stacks or streams in $A$, defined below left. We write $S_A;S_a$ to identify the stack for $a$ in $S_A$. A \emph{state} is a pair $(S_A,\term M)$, and the \emph{transitions} or \emph{steps} are given as top--to--bottom rules below centre. A \emph{run} of the machine is a sequence of steps, written as $(S_A,\term M)\eval(T_A,\term N)$ or with a double line as below right.
\[
\begin{array}{l@{}l}
	S   &~\Coloneqq~\e~\mid~S{\cdot}\term M
\\	S_A &~\Coloneqq~\{\,S_a\,\mid\,a\in A\}
\end{array}
\qquad
\begin{array}{@{(~ }l@{~,~}r@{~)}}
	S_A~;~S_a 	         & \term{[N]a.M}
\\\hline
	S_A~;~S_a{\cdot}\term N &      \term M
\end{array}
\qquad
\begin{array}{@{(~}l@{~,~}r@{~)}}
	S_A~;~S_a{\cdot}\term N &   \term{a<x>.M}
\\\hline
	S_A~;~S_a               & \term{\{N/x\}M}
\end{array}
\qquad
\machine{S_A}M{T_A}N
\]
\end{definition}


\subsection{Encoding effects}

Consider the following $\lambda$-calculus with effects. We will encode it in the poly-$\lambda$-calculus, with its \emph{lazy} or $\cbn$ semantics. At the end of this section we will consider what would be needed to encode its \emph{eager} or $\cbv$ semantics. (We assume familiarity with the operational semantics of $\lambda$-calculus and of effects; for an introduction see e.g.\ Winskel~\cite{Winskel-1993}.)
\[
\begin{array}{r@{~}c@{~}l@{\qquad}l}
	M,N,P & \Coloneqq & x~\mid~M\,N~\mid~\lambda x.M 	& \emph{$\lambda$-calculus}
\\		  & \mid & \rd~\mid~\wrt N;M 					& \emph{input/output}	
\\		  & \mid & c := N;M~\mid~!c  					& \emph{state update and lookup}
\\		  & \mid & N\oplus M ~\mid~ N+M      			& \emph{probabilistic and non-deterministic sum}
\end{array}
\]
The $\cbn$-encoding will follow the description in the introduction. The constructs of the above language are introduced as defined constructs (``sugar'') into the poly-$\lambda$-calculus.

	\textbf{Input/output:}
Input uses a dedicated \emph{input} location $\term\inp\in A$ and is encoded by $\rd\defeq\term{\inp<x>.x}$. The machine is initialized with a stream $S_\inp=\cdots\term{N_3}\cdot\term{N_2}\cdot\term{N_1}$ (infinite to the left), and the pop transition gives the expected operational semantics, below left. 
Writing to output uses a dedicated \emph{output} location $\term\out\in A$ and is encoded by $\wrt N;M\defeq\term{[N]\out.M}$. Evaluation then generates an output stream $\term{N_1},\term{N_2},\dots$ (finite at any step) by the \emph{push} machine transition, below right.
\[
\begin{array}{@{(~}l@{~,~}r@{~)}}
	S_A~;~S_\inp\cdot\term{N} & \term{\inp<x>.x}
\\\hline
	S_A~;~S_\inp              & \term{N}
\end{array}
\qquad
\begin{array}{@{(~}l@{~,~}r@{~)}}
	S_A~;~S_\out                & \term{[N]\out.M}
\\\hline
	S_A~;~S_\out\cdot\term{N} &           \term M 
\end{array}
\]

	\textbf{State:}
A memory cell is modelled by a location $\term c\in A$. The associated stack is expected to hold at most one value, which is preserved by the encoding of the operators, and not enforced externally. Update and lookup are encoded by $\term{c := N ; M}~\defeq~ \term{c<\_>.[N]c.M}$ and $\term{`!c}~\defeq~\term{c<x>.[x]c.x}$ where the underscore $(\term\_)$ represents a variable that does not occur in $\term M$ or $\term N$. In the machine, the stack for each cell is initialized with a (dummy) value, and the transitions then give the expected operational semantics.
\[
\scalebox{0.9}{$
\begin{array}{@{(~}l@{~,~}r@{~)}}
    	 S_A~;~\e_c\cdot \term P & \term{c<\_>.[N]c.M}
\\\hline S_A~;~\e_c              &       \term{[N]c.M}
\\\hline S_A~;~\e_c\cdot \term N &            \term{M}
\end{array}
\qquad
\begin{array}{@{(~}l@{~,~}r@{~)}}
    	 S_A~;~\e_c\cdot \term M &  \term{c<x>.[x]c.x}
\\\hline S_A~;~\e_c              &       \term{[M]c.M}
\\\hline S_A~;~\e_c\cdot \term M &            \term{M}
\end{array}
$}
\]

	\textbf{Probabilistic and non-deterministic sums:}
Following the probabilistic case~\cite{DalLago-Guerrieri-Heijltjes-2020}, probabilistic and non-deterministic sums are included via dedicated locations $\term\rnd,\term\nd\in A$ by $N\oplus M ~\defeq~ \term{\rnd<x>.x\,M\,N}$ and $N+M ~\defeq~ \term{\nd<x>.x\,M\,N}$. The machine is initialized with the corresponding streams of Church-encoded Booleans $\term{lx.ly.x}$ and $\term{lx.ly.y}$, generated probabilistically for $\term\rnd$ and non-deterministically for $\term\nd$. The machine steps are as expected.

\begin{example}
\label{ex:operational}
Consider the following example term and its $\cbn$ encoding in the poly-$\lambda$-calculus. (Numbers can be seen informally as constants, or as Church numerals.)
\[
	\trm{"1a := 2"0 ; ("2 lx "5. `!a)\,"0("3a := 3"0 ;"40"0)}
\quad=\quad
	\trm{"1a<\_>.[2]a "0.["3a<\_>.[3]a"0 . "40"0] . "2 <x>. "5a<y>.[y]a.y}
\]
Its $\cbn$ reduction gives $\term 2$. It evaluates in the machine as follows (where the cell $\term a$ is initialized with zero).
\[
\scalebox{0.85}{$
\begin{array}{@{(~}l@{~;~}l@{~,~}r@{~)}}
    \e_a{\cdot}\term 0 & \e_\lambda                                 &  \trm{"1a<\_>.[2]a "0.["3a<\_>.[3]a"0 . "40"0] . "2 <x>. "5a<y>.[y]a.y}
\\\hline
    \e_a               & \e_\lambda                                 &  \trm{"1[2]a "0.["3a<\_>.[3]a"0 . "40"0] . "2 <x>. "5a<y>.[y]a.y}
\\\hline
    \e_a{\cdot}\term 2 & \e_\lambda                                 &  \trm{["3a<\_>.[3]a"0 . "40"0] . "2 <x>. "5a<y>.[y]a.y}
\\\hline
    \e_a{\cdot}\term 2 & \e_\lambda{\cdot}\trm{"3a<\_>.[3]a"0."40}  &  \trm{"2<x>."5a<y>.[y]a.y}
\\\hline
    \e_a{\cdot}\term 2 & \e_\lambda                                 &  \trm{"5a<y>.[y]a.y}
\\\hline
    \e_a               & \e_\lambda                                 &  \trm{"5["12"5]a."12}
\\\hline
    \e_a{\cdot}\term 2 & \e_\lambda                                 &  \trm{"12}
\end{array}$}
\]
\end{example}


\subsection{Beta-reduction}

In the $\lambda$-calculus, $\beta$-reduction lets successive push- and pop-actions interact. Generalizing to multiple locations, these must be actions on the same stack, while other stacks may be accessed in-between. The $\beta$-rule is then as below, where each $\term{X_i}$ is an abstraction or application not on location $\term a$, and (if the former) not capturing in $\term N$. Reduction is closed under any context. (A formal definition is given in Section~\ref{sec:FMC}.)
\[
	\term{[N]a.X_1\dots X_n.a<x>.M}~\rw~\term{X_1\dots X_n.\{N/x\}M}
\]

\begin{example}
The term from example~\ref{ex:operational} reduces as follows, with reduced redexes underlined.
\begin{center}$
		\trm{"1a<\_>.}\utrm{"1[2]a}\trm{.["3a<\_>.[3]a"0."40"0]."2<x>"0.}\utrm{"5a<y>}\trm{"5.[y]a.y}
~\rw~	\trm{"1a<\_>"0.}\utrm{["3a<\_>.[3]a"0."40"0]."2<x>}\trm{."5["12"5]a."12}
~\rw~	\trm{"1a<\_>"0. "5["12"5]a."12}
~=~		\trm{"{1!50!5}a := 2 "0 ; "1 2}
$\end{center}
\end{example}

Analogous to $\beta$-reduction, $\eta$-reduction is as below, where each $\term{X_i}$ is an abstraction or application not on location $\term a$, and $\term{x}$ does not occur free in any of the $\term{X_i}$ nor in $\term M$. 
\[
	\term{a<x>.X_1\dots X_n.[x]a.M}~\rw_\eta~\term{X_1\dots X_n.M}
\]
As an alternative to these rule schemes, terms may be taken modulo an equivalence $\sim$ generated by the \emph{permutations} below left, and $\beta$- and $\eta$-reduction defined only on adjacent operators, as below right. Observe that the machine semantics immediately validates these equivalences. We will use the below formulation to consider the relation with algebraic effects, but otherwise use the above formulation.
\[
\begin{aligned}
	\trm{"3[M]a."4[N]b."1P} &~\sim~\trm{"4[N]b."3[M]a."1P} \\
	\trm{"3<x>a."4[N]b."1P} &~\sim~\trm{"4[N]b."3<x>a."1P} &&\text{ if }\trm{"3x}\notin\fv{\trm{"4N}} \\
	\trm{"3<x>a."4<y>b."1P} &~\sim~\trm{"4<y>b."3<x>a."1P}
\end{aligned}
\qquad
\begin{aligned}
	\term{[N]a.a<x>.M} &~\rw_\beta~\term{\{N/x\}M}
\\	\term{a<x>.[x]a.M} &~\rw_\eta~\term M \quad &&\text{ if }\term x\notin\fv{\term M}
\end{aligned}
\]

Beta-reduction is confluent, as will be shown more generally for the FMC in Section~\ref{sec:FMC}. This is possible because it follows the algebraic laws for effects~\cite{Plotkin-Power-2002-FOSSACS} instead of their operational semantics. For instance, laws for the interaction of \emph{lookup} and \emph{update} correspond to the following reductions.
\[
\begin{array}{@{}l@{}l@{}l@{}l@{}}
	\trm{"2c:=M;\,"1c:=N;"3P} &~=~ \trm{"2c<\_>.}\utrm{"2[M]c."1c<\_>}\trm{"1.[N]c."3P} & ~\rw~ \trm{"2c<\_>."1[N]c."3P} &~=~ \trm{"5c:=N;"3P}
\\	\trm{"2c:=M;\,"1`!c}      &~=~ \trm{"2c<\_>.}\utrm{"2[M]c."1c<x>}\trm{"1.[x]c.x}    & ~\rw~ \trm{"2c<\_>."1[M]c."5M} &~=~ \trm{"5c:=M;\,M}
\end{array}
\]
The seven algebraic laws for global state of Plotkin and Power~\cite[p.~348]{Plotkin-Power-2002-FOSSACS} arise in our setting from $\beta/\eta$-reduction and $\sim$. Their notation has \emph{update} $u_{\mathit{loc},v}(M)$ of location $\mathit{loc}$ with value $v$ in $M$, and \emph{lookup} $l_{\mathit{loc}}(M)_v$ of the value $v$ at location $\mathit{loc}$ using the value $v$ as a parameter in $M$. These are encoded in the poly-$\lambda$-calculus as below, using abstraction with $x$ instead of parametrization in $v$ in the lookup case. Values $v$ may be taken as arbitrary poly-$\lambda$-terms.
\[
	u_{a,v}(M) ~\defeq~ \term{a<\_>.[v]a.M}
\qquad
	l_a(M)_x   ~\defeq~ \term{a<x>.[x]a.M}
\]

\begin{proposition}
The poly-$\lambda$-calculus with $\rw_{\beta\eta}\cup\sim$ generates the algebraic laws for state.
\end{proposition}

\begin{proof}
By the following equations, where $a\neq b$ and in equation 7, $x\notin\fv v$.
\[
\begin{array}{llllllll}
	1. & l_a(u_{a,y}(x))_y      &~=~& \term{a<y>.[y]a.a<\_>.[y]a.x}      &~\rw_\beta~& \term{a<y>.[y]a.x}                 &~\rw_\eta~& x
\\	2. & l_a(l_a(M_{x,y})_x)_y  &~=~& \term{a<y>.[y]a.a<x>.[x]a.M_{x,y}} &~\rw_\beta~& \term{a<y>.[y]a.\{y/x\}M_{x,y}}    &~=~& l_a(M_{y,y})_y
\\	3. & u_{a,v}(u_{a,v'}(x))   &~=~& \term{a<\_>.[v]a.a<\_>.[v']a.x}    &~\rw_\beta~& \term{a<\_>.[v']a.x}               &~=~& u_{a,v'}(x)
\\	4. & u_{a,v}(l_a(M_x)_x)    &~=~& \term{a<\_>.[v]a.a<x>.[x]a.M_x}    &~\rw_\beta~& \term{a<\_>.[v]a.\{v/x\}M_x}       &~=~& u_{a,v}(M_v)
\\	5. & l_a(l_b(M_{x,y})_y)_x  &~=~& \term{a<x>.[x]a.b<y>.[y]b.M_{x,y}} &~\sim~&      \term{b<y>.[y]b.a<x>.[x]a.M_{x,y}} &~=~& l_b(l_a(M_{x,y})_x)_y 
\\	6. & u_{a,v}(u_{b,v'}(x))   &~=~& \term{a<\_>.[v]a.b<\_>.[v']b.x}    &~\sim~&      \term{b<\_>.[v']b.a<\_>.[v]a.x}    &~=~& u_{b,v'}(u_{a,v}(x))
\\  7. & u_{a,v}(l_b(M_x)_x)    &~=~& \term{a<\_>.[v]a.b<x>.[x]b.M_x}    &~\sim~&      \term{b<x>.[x]b.a<\_>.[v]a.M_x}    &~=~& l_b(u_{a,v}(M_x))_x
\end{array} 
\]
\end{proof}


\subsection{Poly-types}

A simple type for a poly-term represents its expected inputs, taken from multiple independent locations. Correspondingly, the antecedent of an implication is parameterized in a location, and implications on distinct locations may permute.

\begin{definition}
\define{Simple poly-types} are given by the grammar below left, where $\type o$ (omicron) is a \define{base} type and $\type{a(s)\imp t}$ an \define{arrow} type, and considered modulo the congruence $\sim$ given below right.
\[
	\type{r,s,t} ~\Coloneqq~ \type o~\mid~ \type{a(s)\imp t}
\qquad\qquad
	\type{a(r)\imp b(s)\imp t}~~\sim~~\type{b(s)\imp a(r)\imp t} \quad (\type a\neq\type b)~.
\]
The typing rules are as follows, where a \emph{context} $\Gamma$ is a finite function from variables to types.
\[
	\infer{\term{G , x:t |- x:t}}{}
\qquad\qquad
	\infer{\term{G |- a<x>.M : a(s) \imp t}}{\term{G , x:s |- M: t}}
\qquad\qquad
	\infer{\term{G |- [N]a.M : t}}{\term{G |- N : s} && \term{G |- M : a(s)\imp t}}
\]
\end{definition}

Observe that the congruence $\sim$ means that a term $\term{M:a(r)\imp b(s)\imp t}$ may be prefixed by a push action $\term{[N]a}$ where $\term{N:r}$ or by one $\term{[P]b}$ where $\term{P:s}$ (or both, in either order).


\subsection{Towards encoding call--by--value semantics}

The poly-$\lambda$-calculus gives control over when effects are called, as we demonstrate by the following example.

\begin{example}
\label{ex:three args}
Consider the following example term, which is a normal form with $\cbn$ semantics.
\[
	\trm{"1f~"2(c := 2; 0)~"3(`!c)~"4(c := 3;1)}   
\]
With $\cbv$, the arguments may be evaluated left to right, reducing to $\trm{"1f~"20~2~"41}$, or right to left, which gives $\trm{"1f~"20~"43~1}$. The two readings are encoded as follows (using regular applications to $\trm{"1f}$ for readability).
\[
\begin{array}{@{}r@{~}l@{}}
	\trm{"2c<\_>.[2]c"0."3c<x>.[x]c"0."4c<\_>.[3]c"0."1f\,"20\,"3x\,"41}
& 	~\rws~\trm{"2c<\_>"0."4[3]c"0."1f\,"20\,2\,"41}
\\[5pt]
	\trm{"4c<\_>.[3]c"0."3c<x>.[x]c"0."2c<\_>.[2]c"0."1f\,"20\,"3x\,"41}
& 	~\rws~\trm{"4c<\_>"0."2[2]c"0."1f\,"20\,"43\,1}
\end{array}
\]
\end{example}

The encodings rely on repositioning an update $\trm{"4c := 1}$ as a prefix, and for a lookup $\trm{"3`!c}$, on separating the global actions $\trm{"3c<x>.[x]c}$ from the variable $\trm{"3x}$ where the value is used. (The latter idea gives the $\cbv$ semantics in the probabilistic case~\cite{DalLago-Guerrieri-Heijltjes-2020}.) However, it is unlikely that an encoding that only repositions effect operations can encode the $\cbv$ semantics of $\lambda$-calculus with effects. Consider the following example.

\begin{example}
\label{ex:cbv problem}
With a $\cbv$ semantics, the term
\[
	\trm{"1a:= "2(lx. "3b:=1;"2 x)\,"10;"4`!b} 
\]
first reduces the redex, to give $\trm{"1a:= ("3b:=1;"1 0) ;"4`!b}$, which then evaluates by updating $\trm{"3b:=1}$, then $\trm{"1a:=0}$, and reading $\trm{"4`!b}$ as $\trm{"31}$.
To obtain this semantics in the poly-$\lambda$-calculus by manipulaton of effect operations would require lifting $\trm{"3b:=1}$ out of a redex---a process which is likely undecidable in general.
\end{example}

The poly-$\lambda$-calculus thus gives control over effects, but not evaluation behaviour in general. It is an open question whether this is sufficient for practical purposes---one we cannot answer here. Instead, we will consider \emph{sequencing} as a natural way to include $\cbv$ semantics.


\section{The sequential lambda-calculus}
\label{sec:sequencing}

As an instruction sequence for a stack machine, a $\lambda$-term is a string of push and pop actions that must end in a variable. But the machine would naturally accept \emph{any} sequence of actions and variables. Relaxing the variable restriction would further enable composition of sequences. This design of $\lambda$-calculus with sequential composition appears several times in the literature and in practice: as the calculus $\Lambda_s$~\cite{Douence-Fradet-1998}, as the higher-order $\kappa$-calculus~\cite{Power-Thielecke-1999}, and in concatenative programming languages such as Factor~\cite{Pestov-Ehrenberg-Groff-2010}. We call this generalization of the $\lambda$-calculus \emph{sequencing}, and implement it by introducing a \emph{skip} (or \emph{nil}) construct and making the variable a prefix.

\begin{definition}
The \define{sequential $\lambda$-calculus} is given by the following grammar.
\[
\term{M,N,P}
  \quad\Coloneqq\quad \term *
       ~\mid~ \term{x.M}
       ~\mid~ \term{[N].M}
       ~\mid~ \term{<x>.M}
\]
We may omit the trailing $\term{.*}$ from terms for readability. \define{Capture-avoiding composition} $\term{N;M}$ is given by
\[
		\term{*;M}     ~=~ \term M			
\qquad	\term{x.N;M}   ~=~ \term{x.(N;M)}
\qquad	\term{[P].N;M} ~=~ \term{[P].(N;M)}
\qquad	\term{<y>.N;M} ~=~ \term{<y>.(N;M)}
\]
where in the last case $\term y$ is not free in $\term M$. \define{Capture-avoiding substitution} $\term{\{M/x\}N}$ is as follows.
\[
\begin{array}{rcrl@{\qquad}rcrl}
	\term{\{M/x\}*}      &=& \term *
&&	\term{\{M/x\}[P]a.N} &=& \term{[\{M/x\}P]a.\{M/x\}N}
\\	\term{\{M/x\}x.N}    &=& \term{M;\{M/x\}N}	
&&	\term{\{M/x\}a<x>.N} &=& \term{a<x>.N}
\\	\term{\{M/x\}y.N}    &=& \term{y.\{M/x\}N}      & (\term x\neq\term y)
&	\term{\{M/x\}a<y>.N} &=& \term{a<y>.\{M/x\}N} 	& (\term x\neq\term y,~\term y\notin\fv{\term M})
\end{array}
\]
\define{Beta-reduction} is otherwise standard, by closing the rule below left under all contexts. The abstract machine has \define{states} $(S,\term M)$ of a stack and a term, and the transitions below right.
\[
	\term{[N].<x>.M}~\rw~\term{\{N/x\}M}
\qquad\qquad
\begin{array}{@{(~ }l@{~,~}r@{~)}}
	S 	            & \term{[N].M}
\\\hline
	S{\cdot}\term N &     \term M
\end{array}
\qquad
\begin{array}{@{(~}l@{~,~}r@{~)}}
	S{\cdot}\term N &    \term{<x>.M}
\\\hline
	S               & \term{\{N/x\}M}
\end{array}
\]
\end{definition}

\begin{example}
\label{ex:sequential terms}
Consider the following example terms.
\[
	\term{<x>.[x].[x]} 
	\qquad \term{<x>.<y>}
	\qquad \term{[<x>.[x]].<f>.f.f.f}
\]
The first duplicates the top item on the stack; the second removes two items; the third pushes the term $\term{<x>.[x]}$ (which picks up and returns an item), pops it as $\term f$, and runs it three times.
\end{example}

Observe that the changes to evaluation are absorbed by substitution and composition, while $\beta$-reduction and machine evaluation remain unchanged. There is nevertheless a change in perspective from the $\lambda$-calculus, in that the return values or outputs of a computation are pushed to the stack, rather than left as the remainder of the term. Machine evaluation for a term $\term M$ with input stack $S$ is expected to terminate in $\term *$ (just as imperative computation successfully terminates in a \emph{skip} command) with an output stack $T$, i.e.\ $(S,\term M)\eval(T,\term *)$. Then $\term *$ gives the identity run (of zero steps), and $\term{M;N}$ gives composition of runs.

\begin{proposition}
If $(R,\term M)\eval(S,\term*)$ and $(S,\term N)\eval(T,\term*)$ then $(R,\term{M;N})\eval(T,\term*)$.
\end{proposition}


\subsection{Sequential types}

The type system for the sequential $\lambda$-calculus follows that of the $\kappa$-calculus~\cite{Power-Thielecke-1999}; similar type systems have also been studied for stack languages \cite{Stoddart-Knaggs-1993}. The type of a term describes the input/output behaviour of its machine evaluation: it consists of an implication between a vector of input types, one for each element consumed from the stack, and a vector of output types, one for each item returned to the stack.

\begin{definition}
\define{Sequential types} are defined by:
\[
	\type{r,s,t,u}~\Coloneqq~\type{!s>!t} 
\qquad\qquad
	\type{!t}~\Coloneqq~\type{t_1\dots t_n}
\]
Vector concatenation is by juxtaposition, $\type{!s!t}$, and the reverse of a vector  $\type{!t}=\type{t_1\dots t_n}$ is $\type{?t}=\type{t_n\dots t_1}$. Typing rules for the sequential $\lambda$-calculus are given below. A stack is typed by a type vector, where $\Gamma \vdash \e{\cdot}\term{M_1}{\cdots}\term{M_n}:\type{t_1\dots t_n}$ if $\term{G |- M_i : t_i}$ for each $i\leq n$.
\[
	\infer[\TR *]{\term{G |- *:?t>!t}}{}
\quad
	\infer[\TR x]
	 {\term{G , x:?r>!s |- x.M:?r\,?t>!u}}
	 {\term{G , x:?r>!s |- {\phantom{x.}}M:?s\,?t>!u}}
\quad
	\infer[\TR l]
	  {\term{G |- <x>.M : r\,?s>!t}}
	  {\term{G , x:r |- M:?s>!t}}
\quad
	\infer[\TR a]
	  {\term{G |- [N].M: ?s>!t}}
	  {\term{G |- N:r} &
	   \term{G |- M:r\,?s>!t}
	  }
\]
\end{definition}


\begin{example}
The terms in Example~\ref{ex:sequential terms} can be typed as follows.
\[
	\term{<x>.[x].[x] : t>tt} 
\qquad \term{<x>.<y> : ts>} 
\qquad \term{[<x>.[x]].<f>.f.f.f: t>t}
\]
Observe that because stacks are last-in first-out, the identity function on the top two stack items is the term $\term{<x>.<y>.[y].[x] : ts>st}$, whereas the function that swaps them is $\term{<x>.<y>.[x].[y] : ts>ts}$.
\end{example}

\begin{example}
\label{ex:self app}
The term $\term{lx.x\,x}=\term{<x>.[x].x}$ can be typed by assigning $\term x$ a type that does not consume input, i.e.\ of the form $\type{(>!t)}$. The self-application $\term{x\,x}=\term{[x].x}$ then has the type $\type{>(>!t)!t}$, which reflects that the return values accumulate: if evaluating $\term x$ returns the stack $T:\type{!t}$, then $\term{[x].x}$ returns the stack consisting of $\term x$ prepended to $T$. The type derivation is below. Note that the term $\term{(lx.xx)(ly.yy)}$ is not typeable: the argument $\term{ly.yy}$ needs a type that takes input, and then so should $\term x$.
\[
\infer[\TR l]{\vdash~\term{<x>.[x.*].x.* : (>!t)>(>!t)!t}}{
 \infer[\TR a]{\term{x:\!>!t}\;\vdash\;\term{[x.*].x.* : >(>!t)!t}}{
  \infer[\TR x]{\term{x:\!>!t}\;\vdash\;\term{x.*:\!>!t}}{
   \infer[\TR *]{\term{x:\!>!t}\;\vdash\;\term{*:?t>!t}}{}
  }
&\infer[\TR x]{\term{x:\!>!t}\;\vdash\;\term{x.*:(>!t)>(>!t)!t}}{
   \infer[\TR *]{\term{x:\!>!t}\;\vdash\;\term{*:?t(>!t)>(>!t)!t}}{}
}}}
\]
\end{example}

\begin{remark}[Due to Chris Barrett]
\label{rem:type inhabitation}
Observe that all sequential types $\type t$ are inhabited by at least the element $\term{\bot_\tau:t}$, defined below (note that the base case $n=m=0$ gives $\term{*:(>)}$). This is in contrast with the simply-typed $\lambda$-calculus, where not all types are inhabited due to the presence of the uninhabited base type $\type o$.
\[
	\term{\bot_\tau}=\term{<x_n>\dots <x_1>.[\bot_{\tau_1}]\dots[\bot_{\tau_m}]} \quad\text{where}\quad \type t=\type{s_n\dots s_1>t_1\dots t_m}
\]
\end{remark}


\subsection{Encodings of call--by--name calculi}
\label{sec:cbnencodings}

The (regular, call--by--name) $\lambda$-calculus is included as a fragment of the sequential $\lambda$-calculus. We extend this embedding with types and with products, and to Moggi's \emph{computational metalanguage}~\cite{Moggi-1991} following Douence and Fradet~\cite{Douence-Fradet-1998}. The main observation is that implications embed as input-only sequential types, and products as output-only types, below left. The formal, inductive encoding is then below right.
\[
\begin{array}{r@{}l@{\qquad\qquad\qquad}r@{}l@{\qquad}r@{}l}
	\type{t_1\imp\cdots\imp t_n\imp o}&~=~\type{t_1\dots t_n\,>}	& \type{o} &~\defeq~\type{(>)} & \type{r\imp(?s>!t)} &~\defeq~\type{r?s>!t}
\\
	\type{t_1\tim\cdots\tim t_n} &~=~ \type{>\,t_n\dots t_1}		& \type{1} &~\defeq~\type{(>)} & \type{s\tim t}&~\defeq~\type{>\,ts}
\end{array}
\]
Following the types, product terms encode as follows.
\[
	\term{(M,N)} ~\defeq~ \term{[N].[M]}
\qquad\qquad
	\term{\pi_i(P)} ~\defeq~ \term{P;<x_1>.<x_2>.x_i}
\qquad\qquad
	\term{()} ~\defeq~\term*
\]
The computational metalanguage extends the $\lambda$-calculus with monadic type formers $T(\sigma)$, and a \emph{return} construct $[M]_T$ and a \emph{let} construct $\ttlet_T$ parameterized in $T$. In the interpretation in the sequential $\lambda$-calculus, the return value of a monadic function is pushed to the stack. A monadic function type is then interpreted as one with a single output, as below left. The language constructs are encoded as below right. It is easily verified that this extends correctly to type derivations and reductions.
\[
	\type{t_1\imp\cdots\imp t_n\imp T(s)}~=~\type{t_1\dots t_n > s}
\qquad\qquad
	[M]_T~\defeq~\term{[M]}
\qquad	
	\ttlet_T~x\From N~\ttin~M ~\defeq~ \term{N;<x>.M}
\]


\subsection{Encodings of call--by--value calculi}
\label{sec:cbvencodings}

The $\cbv$ $\lambda$-calculus~\cite{Plotkin-1975} and the computational $\lambda$-calculus $\lambda_c$~\cite{Moggi-1989-LICS}, which extends the former with a monadic type constructor $T$ and with the term constructs \emph{return} and \emph{let}, have an encoding in the $\kappa$-calculus~\cite{Douence-Fradet-1998,Power-Thielecke-1999}. We recall this for the sequential $\lambda$-calculus, and observe that it naturally extends to types. The $\cbv$-interpretion of types is naturally viewed in two stages. First, types in isolation are translated as follows.
\[
	o_\val~=~\type{(>)}
\qquad
	(\sigma\imp\tau)_\val~=~\type{s_\val > t_\val}
\qquad
	T(\sigma)_\val~=~\type{>\,s_\val}
\]
Evaluation of a $\lambda_c$-term returns a value, which in the encoding is pushed to the stack. A typed term $M:\tau$ will then translate as $\term{M_\val:\,>\,t_\val}$, with a single output type. Terms of the computational $\lambda$-calculus are then interpreted as below. Again it is easily verified that this extends correctly to type derivations and reduction. The machine behaviour of encoded terms can be seen to follow the SECD-machine~\cite{Landin-1964}.
\[
	x_\val 				~=~ \term{[x]}
\qquad	(\lambda x.M)_\val	~=~ \term{[<x>.M_\val]}
\qquad	([M]_T)_\val		~=~ \term{[M_\val]}
\]
\[
	(M\,N)_\val 		~=~ \term{N_\val;M_\val;<x>.x}
\qquad	(\ttlet_T~x\From N~\ttin~M)_\val ~=~ \term{N_\val;<x>.M_\val}
\]


\subsection{Arrows, call--by--push--value, and kappa-calculus}

The sequential $\lambda$-calculus may encode the related formalisms of Arrows, $\cbpv$, and $\kappa$-calculus. For reasons of space, we will not recall these calculi in detail and only provide an outline to the interested reader.

Hughes's \emph{Arrows}~\cite{Hughes-2000} take the $\lambda$-calculus with products and extend it with a second implication $\sigma\arrimp\tau$, which we interpret directly as that of the sequential $\lambda$-calculus, $\type{s>t}$.
Arrow terms have three constructors, encoded as below. The first lifts a regular term $M{:}\,\rho\imp\sigma$ to an arrow term; the second composes two arrow terms $P\,{:}\,\rho\arrimp\sigma$ and $Q\,{:}\,\sigma\arrimp\tau$; and the third applies the arrow term $P$ to the first element of a pair.
\[
\begin{array}{@{}r@{~}c@{~}lcr@{~}c@{~}l@{}}
	\arr~M &:& \rho\arrimp\sigma &\defeq& \term{<x>.[[x].M]} &:& \type{r>s}
\\	P\ac Q &:& \rho\arrimp\tau   &\defeq& \term{P;Q} &:& \type{r>t}
\\ \first~P &:& (\rho\tim\tau)\arrimp(\sigma\tim\tau) &\defeq& \term{<x>.[x.P]} &:& \type{(>tr)>(>ts)}
\end{array}
\]
The perspective that emerges from this encoding is that the arrow calculus corresponds to a version of the sequential $\lambda$-calculus with binary products instead of stacks (which may be considered $n$-ary products).

Characteristic of $\cbpv$~\cite{Levy-2003,Levy-2006}, and also featured in $\kappa$-calculus, is the separation of \emph{computations} and \emph{values}. In the sequential $\lambda$-calculus this distinction is present, too, if implicitly: values live on the stack, and computations run the machine. To make it explicit, we may extend the calculus with \emph{thunk} and \emph{force} constructs $\term{`!M}$ and $\term{`?V}$, and their reduction rule, as below left. Term constructs of $\cbpv$ (without products or sums) then embed as below right.
\[
\begin{array}{r@{~}c@{~}l}
	    \term{V,W} & ~\Coloneqq~ & \term x ~\mid~ \term{`!M}
\\		\term{M,N} & ~\Coloneqq~ & \term * ~\mid~ \term{`?V.M} ~\mid~ \term{[V].M} ~\mid~ \term{<x>.M}
\\		\term{`!`?N.M} &~\rw~ & \term{N;M}
\end{array}
\qquad\qquad
\begin{array}{rclrcl}
     \ttthunk~M & \defeq & \term{`!M}   &   \ttforce~V & \defeq & \term{`?V}
\\	   \ttret~V & \defeq & \term{[V]}   & N~\ttto~x.~M & \defeq & \term{N;<x>.M} 
\\  \lambda x.M & \defeq & \term{<x>.M} &          V`M & \defeq & \term{[V].M}
\end{array}
\]
Types for $\cbpv$ feature a monadic functor $F$ and a \emph{value} functor $U$. The latter could be introduced into sequential types as below left, though the structure of the arrow type $\type{>}$ makes it redundant. Types then further encode as call--by--name types, below right.
\[
	\type{s,t}~\Coloneqq~\type{!s>!t} \qquad \type{!t}~\Coloneqq~\type{Ut_1\dots Ut_n}
\qquad\qquad
	\rho\imp(\typ{?s>!t})~\defeq~\type{r?s>!t} \qquad F\sigma~\defeq~\type{>s}
\]

The higher-order $\kappa$-calculus~\cite{Power-Thielecke-1999} is closely related to the sequential $\lambda$-calculus. Types are the same as sequential types: an implication between type vectors. Terms omit the unit $\term *$ and have composition $M;N$ as a primitive (rather than prefixing). The remaining constructs encode as follows.
\[
	     \texttt{push}~V    ~\defeq~ \term{[V]} 
\qquad   \kappa x.M         ~\defeq~ \term{<x>.M}
\qquad	 \texttt{mkthunk}~M ~\defeq~ \term{`!M}
\qquad   \texttt{apply}     ~\defeq~ \term{<x>.`?x}
\]


\subsection{String diagrams}

We may view typed sequential $\lambda$-terms as string diagrams. A term $\term{M : r_1\dots r_m>s_n\dots s_1}$ is rendered as below. The wires represent the input and output stacks, with the first element at the top.
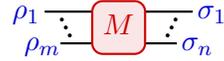
\[
\vc{
\begin{tikzpicture}[termpic]
	\draw
	  (-28,16)node[left]{$\type{r_1}\,$}--(28,16)node[right]{$\,\type{s_1}$}
	  (-22, 4)node[left]{$\type{r_m}$}--(22, 4)node[right]{$\,\type{s_n}$};
	\diagdots[-1.5,3]{-20,10}
	\diagdots[ 1.5,3]{ 20,10}
	\node[termbox] at (0,10) {$\term M$}; 
\end{tikzpicture}}
\]
We will use these diagrams to illustrate how types compose. First, \emph{strict composition} is the composition of terms $\term{M:?r>!s}$ and $\term{N:?s>!t}$ into $\term{M;N:?r>!t}$, below left. This does not give the most general form of composition. For that, we combine it with the following notion of \emph{expansion}. If a term takes an input stack $R$ to an output stack $S$, then when given a larger stack $UR$ it returns $US$, with $U$ untouched. Then if $\term{M:?r>!s}$ also $\term{M:?r?u>!u!s}$, illustrated below right.
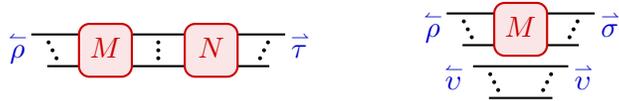
\[
\vc{\begin{tikzpicture}[termpic]
	\wires{?r}{-50}{20}{50}{!t}
	\node[termbox] at (-20,20) {$\term M$};
	\node[termbox] at (20,20) {$\term N$};
	\diagdots[0,3]{0,20}
\end{tikzpicture}}
\qquad\qquad
\vc{\begin{tikzpicture}[termpic]
	\wires{?r}{-30}{20}{30}{!s}
	\wires{?u\,}{-20}{0}{20}{!u}
	\node[termbox] at (0,20) {$\term M$}; 
\end{tikzpicture}}
\]
These constructions combine to give the general case, where in $\term{M;N}$ the type of $\term{M}$ or $\term{N}$ may be expanded.
Note that in the regular $\lambda$-calculus only the first case arises, as the second case requires multiple outputs.
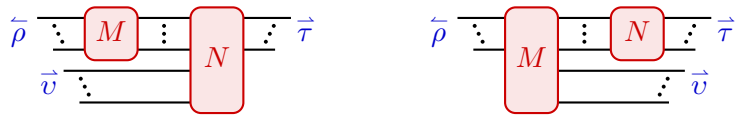
\[
\vc{\begin{tikzpicture}[termpic]
	\wires{?r\,}{-50}{20}{50}{!t}
	\wiresleft{!u}{-40}{0}{20}
	\node[termbox] at (-20,20) {$\term M$};
	\node[tallbox] at ( 20,10) {$\term N$};
	\diagdots[0,3]{0,20}
\end{tikzpicture}}
\qquad\qquad
\vc{\begin{tikzpicture}[termpic]
	\wires{?r\,}{-50}{20}{50}{!t}
	\wiresright{-20}{0}{40}{!u}
	\node[tallbox] at (-20,10) {$\term M$};
	\node[termbox] at ( 20,20) {$\term N$};
	\diagdots[0,3]{0,20}
\end{tikzpicture}}
\] 

\begin{definition}
\label{def:type composition}
\emph{Type composition} $\type{s.t}$ is the partial operation given below, and is undefined otherwise.
\[
\begin{aligned}
	\type{(?r>!s)~.~(?s\,?u>!t)} &~~=~~\type{(?r\,?u>!t)}
\\[2pt]
	\type{(?r>!u\,!s)~.~(?s>!t)} &~~=~~\type{(?r>!u\,!t)}
\end{aligned}
\]
\end{definition}

Observe that the first equation in Definition~\ref{def:type composition} corresponds to the left composition diagram above, and the second equation to the right diagram, where in both cases $\type{!s}$ gives the types of the connecting wires between $\term M$ and $\term N$. The following proposition establishes these basic properties, as well as the familiar \emph{subject reduction}.

\begin{proposition}
\label{prop:type-properties}
Typed terms satisfy the following properties:
\begin{itemize}

	\item Strict composition:
if $~\term{G |- M:?r>!s}$ and $\term{G |- N:?s>!t}$ then $\term{G |- M;N:?r>!t}$.

	\item Expansion:
if $~\term{G |- M:?r>!s}$ then $\term{G |- M:?r\,?u>!u\,!s}$.

	\item Composition:
if $~\term{G |- M : s}$ and $\term{G |- N : t}$ and $\type{s.t}$ is defined, then $\term{G |- M;N : s.t}$.

	\item Substitution:
if $~\term{G |- M:s}$ and $\term{G , x:s |- N:t}$ then $\term{G |- \{M/x\}N:t}$.

	\item Subject reduction:
if $~\term{G |- N:t}$ and $\term N\rw\term M$ then $\term{G |- M:t}$.

\end{itemize}
\end{proposition}


\subsection{Machine termination}

A remarkable aspect of the type system is how it gives a direct connection with termination of the machine. To expose this, we formalize the intuitive meaning of types as describing the initial and final stack of a run of the machine.

\begin{definition}
The set $\run{?s>!t}$ is the set of terms $\term{M}$ such that for any stack $S\in\run{!s}$ there is a stack $T\in\run{!t}$ and a run of the machine $\machine*SMT*$, where $\run{t_1\dots t_n}$ is the set of stacks $\e\cdot\term{N_1}\cdots\term{N_n}$ such that $\term{N_i}\in\run{t_i}$.
\end{definition}

Note that a successful run requires the term to be closed, so a set $\run t$ contains only closed terms. The following lemma shows that $\run t$ is always inhabited by the term $\term{\bot_\tau}$ (see Remark~\ref{rem:type inhabitation}).

\begin{lemma}
For any type $\type t$ the set $\run t$ is inhabited: $\term{\bot_\tau}\in\run t$.
\end{lemma}

If $\term{M:t}$ implies $\term M\in\run t$, then a type derivation \emph{is} a termination proof of the machine. This is Theorem~\ref{thm:run}, and proving it gives a concrete \emph{Tait-style} reducibility proof~\cite{Tait-1967}, where $\run{t}$ takes the r\^ole of the reducibility set for~$\type t$. By using the properties of machine runs the proof is then a simple, direct induction on type derivations.

Vector notation is extended to variables, $\term{!x}=\term{x_1\dots x_n}$, to contexts as $\term{!x:!t}=\term{x_1:t_1,,x_n:t_n}=\Gamma$, and to simultaneous substitutions: if $S=\e\cdot\term{M_1}\cdots\term{M_n}$ then $\term{\{S/{}!x\}}=\term{\{M_1/x_1`,\dots`,M_n/x_n\}}$. We write concatenation of stacks $S$ and $T$ by juxtaposition, $S\,T$.

\begin{lemma}
\label{lem:run}
If $~\term{!w:!w |- M:t}$ then for any $\term{W}\in\run{!w}$, $\term{\{W/{}!w\,\}M}\in\run{t}$.
\end{lemma}

\begin{proof}
By induction on the type derivation. In each case, let $\Gamma=\term{!w:!w}$, let $W$ be a stack in $\run{!w}$, and let $\term{M'}=\term{\{W/{}!w\}M}$.

\begin{itemize}

	\item 
If the derivation is a $\tr*$-rule for $\term{G |- *:?t>!t}$ then there is a trivial zero-step run $\machine*T*T*$.

	\item
If the derivation ends in the variable rule below left, then for any $\term{N}\in\run{!r>!s}$, the inductive hypothesis gives a run for $\term{\{N/x\}M'}$ from any $TS\in\run{!t\,!s}$ to some $U\in\run{!u}$ (second item below). For $\term N$ there is a run from any $R\in\run{!r}$ to some $S\in\run{!s}$ (third item below). These runs compose into one for $\term{\{W/{}!w`,N/x\}x.M}=\term{N;\{N/x\}M'}$ as below right, expanding the stack on the run for $\term N$ by $T$. Note that we may assume $\term x\notin\fv{W}$ (otherwise we rename $\term x$).
\[
	\vc{\infer[\TR x]
	 {\term{G , x:?r>!s |- x.M:?r\,?t>!u}}
	 {\term{G , x:?r>!s |- {\phantom{x.}}M:?s\,?t>!u}}
	}
\qquad	\vc{\machine {TS}{\{N/x\}M'}U*}
\qquad	\vc{\machine RNS*}
\qquad
    \vc{\begin{array}{@{(~}l@{~,~}r@{~)}}
      TR & \term{N;\{N/x\}M'}
    \\\hline\hline\rule[-5pt]{0pt}{15pt}
      TS &   \term{\{N/x\}M'}
    \\\hline\hline\rule[-5pt]{0pt}{15pt}
      U  &           \term*\;
    \end{array}}
\]

	\item
If the derivation ends in the application rule below left, then by the inductive hypothesis for $\term N$ we have $\term{N'=\{W/{}!w\}N}\in\run{r}$, and for $\term M$ we have a run from $\term{M'}$ and any stack $S\in\run{!s}$ with $\term{N'}$ added on top, to some $T\in\run{!t}$. This gives the run for $\term{\{W/{}!w\}[N].M}=\term{[N'].M'}$ below right.
\[
	\vc{\infer[\TR a]
	  {\term{G |- [N].M: ?s>!t}}
	  {\term{G |- N:r} &&
	   \term{G |- M:r\,?s>!t}
	  }}
\qquad\qquad
    \vc{\begin{array}{@{(~}l@{~,~}r@{~)}}
      S               & \term{[N'].M'} 
    \\\hline\rule[-5pt]{0pt}{15pt}
      S\cdot\term{N'} &      \term{M'}
    \\\hline\hline\rule[-5pt]{0pt}{15pt}
      T               &     \term{*}\;
    \end{array}}
\]

	\item
If the derivation ends in the abstraction rule below left, then for any $\term{N}\in\run{r}$ and $S\in\run{!s}$ the inductive hypothesis gives a run for $\term{\{N/x\}M'}$ to some $T\in\run{!t}$. This gives the run for $\term{\{W/{}!w\}<x>.M}=\term{<x>.M'}$ below right.
\[
	\vc{\infer[\TR l]
	  {\term{G |- <x>.M : r\,?s>!t}}
	  {\term{G , x:r |- M:?s>!t}}
	}
\qquad\qquad
    \vc{\begin{array}{@{(~}l@{~,~}r@{~)}}
      S\cdot \term{N} &    \term{<x>.M'}
    \\\hline\rule[-5pt]{0pt}{15pt}
      S               & \term{\{N/x\}M'}
    \\\hline\hline\rule[-5pt]{0pt}{15pt}
      T               &       \term{*}\;
    \end{array}}
\]
\end{itemize}
\end{proof}

\noindent
The following theorem is then immediate.

\begin{theorem}
\label{thm:run}
For a typed, closed term $\term{M:?s>!t}$ and stack $S:\type{!s}$ the machine terminates.
\end{theorem}


\section{The functional machine calculus}
\label{sec:FMC}

The combination of both generalizations, \emph{locations} and \emph{sequencing}, gives the Functional Machine Calculus.

\begin{definition}
\label{def:FMC}
The \define{Functional Machine Calculus} (\FMC) is given by the below grammar, with from left to right the constructors \emph{nil}, a \emph{(sequential) variable}, an \emph{application} or \emph{push action} on the location $\term a$, and an \emph{abstraction} or \emph{pop action} on $\term a$ which binds $\term x$ in $\term M$. Terms are considered modulo $\alpha$-equivalence. 
\[
\term{M,N,P}
  \quad\Coloneqq\quad \term *
       ~\mid~ \term{x.M}
       ~\mid~ \term{[N]a.M}
       ~\mid~ \term{a<x>.M}
\]
\end{definition}

Composition $\term{N;M}$ and substitution $\term{\{N/x\}M}$ are as for the sequential $\lambda$-calculus. The machine is as for the poly-$\lambda$-calculus: a \emph{state} is a pair $(S_A,\term M)$ of a memory and a term, and the \emph{transitions} are:
\[
\begin{array}{@{(~ }l@{~,~}r@{~)}}
	S_A~;~S_a 	         & \term{[N]a.M}
\\\hline
	S_A~;~S_a{\cdot}\term N &      \term M
\end{array}
\qquad
\begin{array}{@{(~}l@{~,~}r@{~)}}
	S_A~;~S_a{\cdot}\term N &   \term{a<x>.M}
\\\hline
	S_A~;~S_a               & \term{\{N/x\}M}
\end{array}
\]
Beta-reduction is as for the poly-$\lambda$-calculus: a redex consists of a successive application and abstraction \emph{on the same location}, separated by any number of actions on other locations. We will now make this formal.

\define{Head contexts} $\term H$ are defined as below left. The term obtained by replacing the hole $\term{\{\}}$ in $\term{H}$ with $\term M$ is denoted $\term{H.M}$, where a binder $\term{a<x>}$ in $\term{H}$ captures in $\term M$. The \emph{binding variables} $\bv{\term H}$ of $\term H$ are those variables $\term x$ where $\term H$ is constructed over $\term{a<x>}$. The set of \define{locations} used in a term or context is denoted $\loc M$ respectively $\loc H$. Then 
\define{Beta-reduction} is defined by the rewrite rule schema below right, where $\term a\notin\loc H$ and $\bv{\term H}\cap\fv{\term N}=\varnothing$, and is closed under all contexts.
\[
	\term{H} ~\Coloneqq~ \term{\{\}}~\mid~\term{[M]a.H}~\mid~\term{a<x>.H}
\qquad\qquad
	\term{[N]a.H.a<x>.M}~\rw~\term{H.\{N/x\}M}
\]
We will first consider the untyped calculus. We give the $\cbv$ encoding of effects and provide an intuition for programming in the FMC, then establish confluence and connect machine evaluation to $\beta$-reduction. We then consider simple types. Constants will be used informally, in examples.


\subsection{Call--by--value with effects}

We extend the encoding $(-)_\val$ of the computational $\lambda$-calculus of Section~\ref{sec:cbvencodings} to effects as follows. (The case for $N\oplus M$ is the same as for $\cbn$, as it expects a Church boolean for $\term x$.)
\[
\begin{aligned}
            \rd_\val &~=~ \term{\inp<x>.[x]}
\\   (\wrt N;M)_\val &~=~ \term{N_\val;<x>.[x]\out.M_\val}
\\   (c := N;M)_\val &~=~ \term{N_\val ; <x>.c<\_>.[x]c.M_\val}
\\	         !c_\val &~=~ \term{c<x>.[x]c.[x]}
\\	(N\oplus M)_\val &~=~ \term{\rnd<x>.[M_\val].[N_\val].x}
\end{aligned}
\]

\begin{example}
The term from Example~\ref{ex:operational}, below, is given a $\cbv$ interpretation as follows.
\[
	\trm{"1a := 2"0 ; ("2 lx "5. `!a)\,"0("3a := 3"0 ;"4 0"0)}~\rw_\cbv~\trm{"33}
\]
 Integers are values, and the translation will use $i_\val=\term{[i]}$. An update with an integer then simplifies by:
\[
\begin{array}{r@{}r@{}r}
	(\term{a := i;M})_\val &~=~ \uterm{[i].<x>}\term{.a<\_>.[x]a;M} &~\rw~ \term{a<\_>.[i]a;M}
\end{array}
\]
The $\cbv$-translated term, after applying this reduction to the two updates, further reduces as follows.
\[
	\trm{"1a<\_>."1[2]a "0 . "3 a<\_>"3.[3]a "0 . "4[0]"0 . ["2<x>. "5 a<y>.[y]a.[y]"0] . <z> . z}
	~\rws~ 
	\trm{"1a<\_>"0."5 ["33"5]a.["33"5]}
\]
\end{example}

\begin{example}
The term from Example~\ref{ex:cbv problem},
\[
	\trm{"1a:= "2(lx. "3b:=1;"2 x)\,"10;"4`!b} ~\rws_\cbv~ \trm{"31}
\]
translates and reduces as follows (with the same shortcut for update as above).
\[
\begin{array}{@{}l@{}}
	\trm{"1[0].}\utrm{"2[<x>."3b<\_>.[1]b"0."2[x]].<y>}\trm{"2.y."1<z>.a<\_>.[z]a."4b<u>.[u]b.[u]}
\\ \quad\rw~  \utrm{"1[0]."2<x>}\trm{."3b<\_>."3[1]b}.\,\utrm{"2[x]."1<z>}\trm{"1.a<\_>.[z]a."4b<u>.[u]b.[u]}
\\ \quad\rws~ \trm{"3b<\_>.}\utrm{"3[1]b}\trm{."1a<\_>.[0]a"0.}\utrm{"4b<u>}\trm{"4.[u]b.[u]}
\\ \quad\rw ~ \trm{"3b<\_>"0."1a<\_>.[0]a"0."4["31"4]b.["31"4]}
\\ \quad\sim~ \trm{"1a:=0;"{3!50!4}b:=1;1} 
\end{array}
\]
\end{example}


\subsection{Programming in the FMC}

As in the sequential $\lambda$-calculus, programming in the FMC naturally follows the concatenative paradigm. A term $\term M$ is viewed as a function taking an input memory $S_A$ to an output memory $T_A$ by a run of the machine $(S_A,\term M)\eval(T_A,\term *)$. Functions standardly operate on the main stack $\term l$, and it is then natural to consider effect operators that transfer values between the main stack and other locations, as the $\cbv$ translations of effect operators do. We introduce the following operations for input and output, a random generator, and a memory cell $\term c$. We further add \emph{definitions} or \emph{let}, as a redex.
\[
\begin{array}{r@{~}l}
	\term{\print} &~\defeq~\term{<x>.[x]\out}
\\	\term{\rd}    &~\defeq~\term{\inp<x>.[x]}
\\	\term{\rand}  &~\defeq~\term{\rnd<x>.[x]}
\end{array}
\qquad
\begin{array}{r@{~}l}
	\term{\get c} &~\defeq~\term{c<x>.[x]c.[x]}
\\	\term{\set c} &~\defeq~\term{<x>.c<\_>.[x]c}
\\  \term{(x=N);M}  &~\defeq~\term{[N].<x>.M}
\end{array}
\]
Constant operations such as the conditional $\term\ifthen$ pop the required number of items from the main stack, and reinstate their result, as is standard for stack languages. For example:
\[
\begin{array}{c@{\qquad\quad}c}
	\step{S_A;S_\lambda\cdot\term 2\cdot\term 3}{+.M}{S_A;S_\lambda\cdot\term 5}M
&	\step{S_A;S_\lambda\cdot\term P\cdot\term N\cdot\term\bot}{\ifthen.M}{S_A;S_\lambda\cdot\term P}M
\end{array}
\]
The FMC then operates similarly to a stack calculus for arithmetic: an expression $1 + ((2 + 3) \times 4)$ is given as a term $\term{[4].[3].[2].+.\times.[1].+}$ which indeed returns $21$. This results in an imperative programming style similar to Haskell's \emph{do}-notation, with the difference that terms may have \emph{any} number of return values, and consume any number of previously returned values.

\begin{example}
\label{ex:ff}
Consider the following example, where $\term\rnd$ is taken to randomly sample natural numbers.
\[
	\trm{("1f~=~"2\rand"0~;~"3\set c"0~;~"4\get c"0)~;~"1f"0~;~"1f"0~;~"5+"0~;~"6\print}
\]
The term assigns $\term f$ to be the function that draws a random number, stores it in cell $\term c$, and reads the value at $\term c$ again as its return value. It then executes $\term f$ twice, sums the results, and sends that to output. The overall actions should be to take two random inputs $i$ and $j$, to update the cell $\term c$ with the last value $j$, and to output $i+j$. In Figure~\ref{fig:ex:ff} the term is first interpreted as an FMC-term and reduced to normal form, where each line is a beta-step, and then evaluated on the machine, where the initial memory provides the two expected inputs on $\term\rnd$ and one on $\term c$. (For compactness we give locations as a header and show only necessary stack elements.)
\end{example}

\begin{figure}
\[
\scalebox{0.85}{$
\begin{array}{@{}rl@{}}
                 \trm{["2\rnd<x>.}\utrm{"2[x]"0."3<y>}\trm{"3.c<\_>.[y]c"0."4c<z>.[z]c.[z]"0]."1<f>"0."1f"0."1f"0."5+"0."6<p>.[p]out}
\\[5pt]  \trm{["2\rnd<x>"0."3c<\_>.}\utrm{"3["2x"3]c"0."4c<z>}\trm{"4.[z]c.[z]"0]."1<f>"0."1f"0."1f"0."5+"0."6<p>.[p]out}
\\[5pt]  \utrm{["2\rnd<x>"0."3c<\_>"0.~"4["2x"4]c.["2x"4]"0]."1<f>}\trm{."1f"0."1f"0."5+"0."6<p>.[p]out}
\\[5pt]  \trm{"2\rnd<x>"0."3c<\_>"0.}\utrm{"4["2x"4]c}\trm{."4["2x"4]"0~~.~"2\rnd<y>"0.}\utrm{"3c<\_>}\trm{."4["2y"4]c.["2y"4]"0~~."5+"0."6<p>.[p]out}
\\[5pt]  \trm{"2\rnd<x>"0."3c<\_>"0."4["2x"4]"0~~.~"2\rnd<y>"0."4["2y"4]c.["2y"4]"0~~."5+"0."6<p>.[p]out}
\end{array}$}
\]
\bigskip
\[
\scalebox{0.85}{$
\begin{array}{@{(~}l@{~;~}l@{~;~}l@{~;~}l@{~,~}r@{~)}}
	\multicolumn{1}{@{}c@{}}{\out} 
  & \multicolumn{1}{@{}l@{}}{\rnd}
  & \multicolumn{1}{@{}l@{}}{\,c}
  & \multicolumn{1}{@{}l@{}}{\hspace{5pt}\lambda}
\\[2pt]             & \trm{"26}\cdot\trm{"27} & \trm{"3*} &                         & \trm{"2\rnd<x>"0."3c<\_>"0."4["2x"4]"0."2\rnd<y>"0."4["2y"4]c.["2y"4]"0."5+"0."6<p>.[p]out}
\\\hline            & \trm{"26}               & \trm{"3*} &                         &             \trm{"3c<\_>"0."4["27"4]"0."2\rnd<y>"0."4["2y"4]c.["2y"4]"0."5+"0."6<p>.[p]out}
\\\hline            & \trm{"26}               &           &                         &                       \trm{"4["27"4]"0."2\rnd<y>"0."4["2y"4]c.["2y"4]"0."5+"0."6<p>.[p]out}
\\\hline            & \trm{"26}               &           & \trm{"27}               &                                   \trm{"2\rnd<y>"0."4["2y"4]c.["2y"4]"0."5+"0."6<p>.[p]out}
\\\hline            &                         &           & \trm{"27}               &                                               \trm{"4["26"4]c.["26"4]"0."5+"0."6<p>.[p]out}
\\\hline            &                         & \trm{"26} & \trm{"27}               &                                                        \trm{"4["26"4]"0."5+"0."6<p>.[p]out}
\\\hline            &                         & \trm{"26} & \trm{"27}\cdot\trm{"26} &                                                                    \trm{"5+"0."6<p>.[p]out}
\\\hline            &                         & \trm{"26} & \trm{"513}              &                                                                          \trm{"6<p>.[p]out}
\\\hline            &                         & \trm{"26} &                         &                                                                         \trm{"6["513"6]out}
\\\hline \trm{"513} &                         & \trm{"26} &                         &                                                                                     \trm{*}
\end{array}$}
\]
\caption{Reduction followed by machine evaluation of the term in Example~\ref{ex:ff}}
\label{fig:ex:ff}
\end{figure}


\subsection{Confluence}
\label{sec:confluence}

In demonstrating confluence for the $\lambda$-calculus, the difficulty is reduction inside an argument: duplicating or deleting it creates converging reductions of different length. By contrast, \emph{spine reduction}, which reduces in every context \emph{except} argument position, is \emph{diamond} (peaks converge in one step).

While the two extensions of the FMC, \emph{locations} and \emph{sequencing}, do create new configurations of overlapping redexes, the situation is fundamentally the same. Spine reduction, defined analogously as reduction in every context except argument position, is diamond, and the remaining problem is the same as for the $\lambda$-calculus. The calculus thus remains confluent, which will be proved by the standard  \emph{parallel reduction} technique of Tait and Martin-L\"of~(see \cite{Barendregt-1984}) and Takahashi~\cite{Takahashi-1995}. 

The new configurations are the following. \emph{Sequencing} introduces terms of the form $\term{N;x.M}$, with reduction in $\term M$ or in $\term N$, where the latter may induce substitutions in $\term{x.M}$. But reduction in $\term N$ may not duplicate a redex in $\term M$, and vice versa, so a peak of this kind converges immediately. 

\emph{Locations} create two new overlapping configurations,
\[
\begin{aligned}
\text{nested:}      \qquad & \trm{"3[N]a"0."4[P]b"0."4b<y>"0."3a<x>"0."1M} \\[5pt]
\text{interleaved:} \qquad & \trm{"3[N]a"0."4[P]b"0."3a<x>"0."4b<y>"0."1M} 
\end{aligned}
\]
but both resolve immediately: in each case the two reducts will converge in one step to $\trm{"3\{N/x\}"4\{P/y\}"1M}$.

Because of this, the problem of confluence amounts to the problem of reduction in arguments, which may be duplicated or deleted, as it does in the regular $\lambda$-calculus. We formalize this observation in the following proposition (which is independent of the confluence result). \emph{Spine reduction} is given by closing the $\beta$-step under all contexts except in argument position. That is, reduction in each of $\term{x.M}$, $\term{[N]a.M}$, and $\term{<x>.M}$ may take place in $\term M$, but not in $\term N$.

\begin{proposition}
Spine reduction is diamond.
\end{proposition}

The full confluence proof is a standard application of parallel reduction, and is given in Appendix~\ref{A:confluence}.

\begin{theorem}
\label{thm:confluence}
Reduction $\rw$ is confluent.
\end{theorem}


\subsection{Simple types for the FMC}
\label{sec:FMC types}

As with sequential types, FMC types will represent the input/output behaviour of the machine. Since a memory is a family of stacks, types will use families of type vectors.

\begin{definition}
\emph{FMC-types} $\type{r,s,t,u}$ are given by:
\[
		\type t	   ~\Coloneqq~ \type{?sA>!tA}
\qquad	\type{!tA} ~\Coloneqq~ \{\type{!t_a}\mid a\in A\}
\qquad	\type{!t}  ~\Coloneqq~ \type{t_1\dots t_n}
\]
The typing rules for the FMC are in Figure~\ref{fig:FMC-types}. They use the following notation.
\emph{Concatenation} of two vector families is pointwise, $\type{!sA !tA}=\{\type{!s_a!t_a}\mid a\in A\}$, and he empty vector is denoted $\type\e$. A \emph{slice} $\type{!tA|_a}$ of a family $\type{!tA}=\{\type{!t_a}\mid a\in A\}$ is the vector $\type{!t_a}$, and a slice $\type{t_a}$ of a type $\type t=\type{?rA>!sA}$ is the type $\type{?r|_a~>~!s|_a}$ (i.e.
$\type{t_a}$ restricts $\type t$ to a single location $a$). \emph{Composition} is slice-wise:  $\type{s.t}~=~\{ \type{s|_a.t|_a}\mid a\in A\}$. Finally, a \emph{singleton} $\type{a(!t)}$ is a type vector at a single location, with all other locations empty, defined by $\type{a(!t)|_a}=\type{!t}$ and $\type{a(!t)|_b}=\type{\e}$ for $\type a\neq \type b$. A singleton $\type{\lambda(!t)}$ on the main location $\lambda$ may be written as $\type{!t}$.
\end{definition}


\begin{figure}
\[
\begin{array}{c@{\qquad}c}
	\infer[\TR *]{\term{G |- *:?tA>!tA}}{}
&
 	\infer[\TR a]
	  {\term{G |- [N]a.M: ?sA>!tA}}
	  {\term{G |- N:r} &&
	   \term{G |- M:a(r)\,?sA>!tA}
	  }
\\ \\
	\infer[\TR x]
	 {\term{G , x:?rA>!sA |- x.M:?rA\,?tA>!uA}}
	 {\term{G , x:?rA>!sA |- {\phantom{x.}}M:?sA\,?tA>!uA}}
&
	\infer[\TR l]
	  {\term{G |- a<x>.M : a(r)\,?sA>!tA}}
	  {\term{G , x:r |- M:?sA>!tA}}
\end{array}
\]
\caption{Typing rules for the Functional Machine Calculus}
\label{fig:FMC-types}
\end{figure}


Poly-types embed as types of the form $\type{?tA>}$ by the definitions $\type o~\defeq~\type{(>)}$ and $\type{a(s)\imp(?tA>)}~\defeq~\type{a(s)\,?tA>}$, and sequential types embed directly as types over only the location $\lambda$. The properties proved of sequential types in Section~\ref{sec:sequencing} carry over straightforwardly: strict composition, expansion, composition, substitution, subject reduction, and termination of the machine.

\begin{example}
The \emph{singleton} construct $\type{a(!t)}$ gives a natural way of writing types in practice. The term from Example~\ref{ex:ff} may be typed as follows, where $\type\Z$ is a base type of integers.
\[
	\term{(f~=~\rand~;~\set c~;~\get c)~;~f~;~f~;~+~;~\print~:~\rnd(\Z~\Z)~c(\Z) > c(\Z)~\out(\Z)}
\]
The type expresses that the term pops two integers from $\type\rnd$ and one from $\type c$, and pushes one integer to $\type c$ and one to $\type\out$. (Note that there are other ways of writing the same type, since singleton types on different locations may permute.) Below left are the types of the defined subterms, and the full type is built up below right by composing these.
\[
\begin{array}{r@{}r@{}l}
                    & \term{\plus}          &~:~ \type{\Z~\Z>\Z}
\\ \term{\rand}  ~=~& \term{\rnd<x>.[x]}    &~:~\type{\rnd(\Z)>\Z}
\\ \term{\print} ~=~& \term{<x>.[x]\out}    &~:~\type{\Z>\out(\Z)}
\\ \term{\set c} ~=~& \term{<x>.c<\_>.[x]c} &~:~\type{\Z~c(\Z)>c(\Z)}
\\ \term{\get c} ~=~& \term{c<x>.[x]c.[x]}  &~:~\type{c(\Z)>c(\Z)~\Z}
\end{array}
\qquad\qquad
\begin{array}{@{}l@{}l@{}}
	\phantom{(f{=}}\term{\rand;\set c} &~:~\type{\rnd(\Z)~c(\Z)~>~c(\Z)}
\\ 
	\phantom{(f{=}}\term{\rand;\set c;\get c} &~:~\type{\rnd(\Z)~c(\Z)~>~c(\Z)~\Z}
\\ 
	\term{(f=\rand;\set c;\get c)} &~:~\type{(>)}
\\ 
	\term{(f=\dots);f;f} &~:~\type{\rnd(\Z~\Z)~c(\Z)~>~c(\Z)~\Z~\Z}
\\ 
	\term{(f=\dots);f;f;+;\print} &~:~\type{\rnd(\Z~\Z)~c(\Z)~>~c(\Z)~\out(\Z)}
\end{array}
\]
The third term is only a redex, which pushes and then pops a value, giving it an empty type $\type{(>)}$. The type of the fourth term is that of $\term{f;f}$. Separating the self-composition of the type of $\term f$ for each location exposes how the inputs on $\type\rnd$ and outputs on $\type\lambda$ accumulate, while the output and input on $\type c$ interact:
\[
\begin{array}{r@{\,}c@{\,}lcr@{\,}c@{\,}lcr@{\,}c@{\,}l}
			\type{\rnd(\Z)} & \type> &              & \type. & \type{\rnd(\Z)} & \type> &              & = & \type{\rnd(\Z~\Z)} & \type> &
\\[-3pt]	\type{c(\Z)}    & \type> & \type{c(\Z)} & \type. & \type{c(\Z)}    & \type> & \type{c(\Z)} & = & \type{c(\Z)}       & \type> & \type{c(\Z)} 
\\[-3pt]	                & \type> & \type{\Z}    & \type. &                 & \type> & \type{\Z}    & = &                    & \type> & \type{\Z~\Z} 
\end{array}
\]
By way of illustration, in Haskell the same example may be written as follows.

\begin{center}
\texttt{\small
\begin{tabular}{@{}l}
  {\color2example} :: {\color3RandT StdGen (StateT Int IO) ()}
\\[-3pt]{\color2example} = {\color1do}
\\[-3pt]\quad{\color1let} f = {\color1do} x <- rand; lift (put x); lift get
\\[-3pt]\quad y <- f
\\[-3pt]\quad z <- f
\\[-3pt]\quad lift (lift (print (y+z)))
\end{tabular}
}
\end{center}

To combine the effects of I/O, state, and random generation, the Haskell example uses a stack of monad transformers. The effects are then layered in a fixed order, and to access each effect the function \texttt{lift} is used to move to the next layer. Thus, the \texttt{rand} action requires no lifting since the random generator is at the top of the transformer stack; the state actions \texttt{put} and \texttt{get} must be lifted once; and the \texttt{print} action must be lifted twice, since I/O is at the bottom of the transformer stack.
\end{example}

The example highlights the following differences between monad transformers and the FMC. Firstly, reader/writer effects in the FMC combine seamlessly, without requiring their organisation in a stack, and are accessed by their locations instead of a lifting function. Secondly, sequential composition in the FMC is a basic operation, with a close connection between the syntax and its execution, where Haskell \texttt{do}-notation is syntactic sugar for more involved monadic operations. Thirdly, the FMC type system accounts for the effectful operations that a term performs, where monad transformer types indicate which effects may be present, but not which operations they perform.

Of course, where monads are universal, the FMC is presently restricted to reader/writer effects. How to broaden the range of effects covered in the FMC is an important direction for future work.


\section{Further work}

We have given an exploratory overview of the Functional Machine Calculus with the most essential results: the natural capture of algebraic laws for effects by reductions and permutations; the encoding of related formalisms for controlling execution behaviour such as monadic constructs, $\cbpv$, $\kappa$-calculus, and Arrows; confluence; and termination of the machine with simple types. A forthcoming paper will strengthen these results with domain-theoretic and categorical semantics, and strong normalization with simple types.

Present and future work aims to extend the FMC beyond reader/writer effects and with standard features. One direction is to include sum types, datatypes, and error handling, where it looks possible to capture all three in a uniform way. A second direction is to introduce parallel composition and explore the relation with process calculi, where our type system promises to give something closely related to session types. A third direction is local mutable store, by introducing a \emph{new} construct for locations, and generalizing locations to \emph{regions} to capture mutable data structures (arrays, graphs), which then leads naturally into an exploration of dependent types for the FMC. A fourth direction is to explore the connection with string diagrams, and to introduce constructs to capture diagrammatic reasoning, for example interaction nets or quantum diagrammatic systems.

An important theoretical challenge is type inference. This appears to be open also for the sequential $\lambda$-calculus: existing algorithms for concatenative languages are limited~\cite{Stoddart-Knaggs-1993,Diggins-2008}, and (we believe) unable to find the type for the self-application in Example~\ref{ex:self app}.


\section*{Acknowledgements}

I am deeply grateful to Chris Barrett and Guy McCusker for valuable contributions to both presentation and content of this work. 
I would further like to thank the following people for constructive discussions and pointers to the literature: Ugo Dal Lago, Giulio Guerrieri, Paul Levy, John Power, Alex Simpson, and Vincent van Oostrom. This work was supported by EPSRC Project EP/R029121/1 \emph{Typed lambda-calculi with sharing and unsharing}.


\bibliographystyle{entics}
\bibliography{newFMC}

%
\appendix

\section{The confluence proof}
\label{A:confluence}

This section gives a complete proof of confluence using parallel reduction. The particular approach is to extend the syntax of terms with marked redexes, to define parallel reduction and to identify residuals.

\begin{definition}
A \define{marked redex} is one whose application and abstraction are marked by the symbols $(\MXL)$ and ($\MXR$), as follows.
\[
	\term{[N]a\MXL.H.\MXR a<x>.M}~.
\]
A \define{redex-marked} or \define{marked} term $\term{M^\MX}$ is one where a selection of redexes is marked. For a marked term consisting of a head context and a term $\term{(H.M)^\MX}$, the marking of both components separately is indicated by $\term{H^\MXL}$ and $\term{M^\MXR}$. The \define{marked reduct} $\term{(M^\MX)_\MX}$ of a marked term $\term{M^\MX}$ is defined as follows.
\[
\begin{aligned}
		\term{(*)_\MX}              &= \term *				 
	  & \term{([N^\MX]a.M^\MX)_\MX} &= \term{[(N^\MX)_\MX]a.(M^\MX)_\MX} 
\\[5pt] \term{(x.M^\MX)_\MX}        &= \term{x.(M^\MX)_\MX} 
      & \term{(a<x>.M^\MX)_\MX}     &= \term{a<x>.(M^\MX)_\MX}
\end{aligned}
\]
\[
	\term{([N^\MX]a\MXL.H^\MXL.\MXR a<x>.M^\MXR)_\MX} 
	=
	\term{\{(N^\MX)_\MX/x\}(H^\MXL.M^\MXR)_\MX}
\]
A \define{parallel reduction step} $\term{M^\MX}\rwp \term{(M^\MX)_\MX}$ relates a marked term to its marked reduct, and an unmarked term to all its marked reducts: $\term M\rwp \term{(M^\MX)_\MX}$ for each marking $\term{M^\MX}$ of $\term{M}$.
\end{definition}

To reduce clutter, a marked reduct $\term{(M^\MX)_\MX}$ may be abbreviated as $\term{M_\MX}$. 

\begin{proposition}
\label{prop:parallel correct}
A parallel reduction step is a beta-reduction: $\term{M}\rws\term{(M^\MX)_\MX}$ whenever $\term{M^\MX}$ is a marking of $\term M$.
\end{proposition}

\begin{proof}
By induction on the size of $\term{M^\MX}$.
\begin{itemize}
	\item Unit case: immediate by $\term{(*)_\MX}=\term*$.
	\item Variable case: if $\term{M^\MX}\rws\term{M_\MX}$ then $\term{x.M^\MX}\rws\term{x.M_\MX}$.
	\item Unmarked application case: if $\term{M^\MX}\rws\term{M_\MX}$ and $\term{N^\MX}\rws\term{N_\MX}$ then $\term{[N^\MX]a.M^\MX} ~\rws~ \term{[N_\MX]a.M_\MX}$.
	\item Unmarked abstraction case: if $\term{M^\MX}\rws\term{M_\MX}$ then $\term{a<x>.M^\MX}\rws\term{a<x>.M_\MX}$.
	\item Redex case: if $\term{(H.M)^\MX}\rws\term{(H.M)_\MX}$ and $\term{N^\MX}\rws\term{N_\MX}$ then
	\[
	\begin{array}{r@{}c@{}l@{}}
    	\term{[N^\MX]a\MXL.H^\MXL.\MXR a<x>.M^\MXR} 
    	&~\rws~& \term{[N_\MX]a\MXL.H^\MXL.\MXR a<x>.M^\MXR}
	\\	&~\rw~ & \term{H^\MXL.\{N_\MX/x\}M^\MXR}
	\\	&~=~   & \term{\{N_\MX/x\}(H^\MXL.M^\MXR)}
	\\	&~\rws~& \term{\{N_\MX/x\}(H.M)_\MX}~.
	\end{array}
	\]
where in the second line, $a\in\loc{\term H}$ by the definition of a redex, and for the equality on the third line, $\term H$ does not bind in $\term N$ and $\term x$ is not free in $\term H$, by $\alpha$-equivalence. 
\end{itemize}
\end{proof}

To show that parallel reduction is diamond, a term is reduced according to two markings, $\term{\MX}$ and $\term{\MO}$. These are then applied simultaneously and interchangeably, and may be considered a single marking $\termcolor{{\MX}{\MO}}=\termcolor{{\MO}{\MX}}$~:
\[
	 \term{(M^\MX)^\MO} = \term{M^{\MX\MO}} = \term{(M^\MO)^\MX}
\]
If this term is reduced relative to one marking, the other is preserved: $\term{(M^{\MX\MO})_\MO}=\term{((M^\MO)_\MO)^\MX}$. The proof then amounts to showing that reducing by each marking is commutative and the same as reducing along both markings simultaneously:
\[
	\term{(M_\MX)_\MO} = \term{M_{\MX\MO}} = \term{(M_\MO)_\MX}
\]

\begin{lemma}
\label{lem:parallel-composition}
Parallel reduction commutes with composition: $\term{M_\MX;N_\MX}=\term{(M;N)_\MX}$.
\end{lemma}

\begin{proof}
By induction on the size of $\term{M}$.
\begin{itemize}

	\item
Unit case: $\term{(*)_\MX;N_\MX}=\term{*;N_\MX}=\term{N_\MX}=\term{(*;N)_\MX}$.

	\item 
Variable case: if $\term{M_\MX;N_\MX}=\term{(M;N)_\MX}$ then
\[
	\term{(x.M)_\MX;N_\MX}=\term{x.(M_\MX;N_\MX)} = \term{x.(M;N)_\MX} = \term{((x.M);N)_\MX}~.
\]

	\item 
Unmarked application case: if $\term{M_\MX;N_\MX}=\term{(M;N)_\MX}$ then
\[
\begin{array}{r@{}l}
	\term{([P]a.M)_\MX;N_\MX}
		&~=~\term{[P_\MX]a.(M_\MX;N_\MX)}
	\\  &~=~\term{[P_\MX]a.(M;N)_\MX}
	\\  &~=~\term{(([P]a.M);N)_\MX}~.
\end{array}
\]

	\item 
Unmarked abstraction case: if $\term{M_\MX;N_\MX}=\term{(M;N)_\MX}$ then
\[
\begin{array}{r@{}l}
	\term{(a<x>.M)_\MX;N_\MX}
		&~=~\term{a<x>.(M_\MX;N_\MX)}
	\\  &~=~\term{a<x>.(M;N)_\MX}
	\\  &~=~\term{((a<x>.M);N)_\MX}~.
\end{array}
\]

	\item 
Marked redex case: if $\term{(H.M)_\MX;N_\MX}=\term{((H.M);N)_\MX}$ and $\term x\notin\fv{\term N}$ then
	\[
	\begin{array}{r@{}l}
    	\term{([P]a\MXL.H.\MXR a<x>.M)_\MX;N_\MX} 
    	&~=~ \term{(\{P_\MX/x\}(H.M)_\MX);N_\MX }
	\\	&~=~ \term{\{P_\MX/x\}((H.M);N)_\MX }
	\\	&~=~ \term{(([P]a\MXL.H.\MXR a<x>.M);N)_\MX}~.
	\end{array}
	\]
\end{itemize}
\end{proof}

\begin{lemma}
\label{lem:parallel-substitution}
Parallel reduction commutes with substitution: $\term{\{N_\MX/x\}M_\MX}=\term{(\{N/x\}M)_\MX}$.
\end{lemma}

\begin{proof}
By induction on the size of $\term{M^\MX}$.
\begin{itemize}
	\item 
Unit case: immediate by $\term{\{N/x\}*}~=~\term*$.

	\item 
Variable case: if $\term{\{N_\MX/x\}M_\MX}=\term{(\{N/x\}M)_\MX}$ then by Lemma~\ref{lem:parallel-composition},
\[
\begin{array}{r@{}l}
	\term{\{N_\MX/x\}(x.M)_\MX}
	&~=~ \term{\{N_\MX/x\}(x.M_\MX)}
\\	&~=~ \term{N_\MX;\{N_\MX/x\}M_\MX}
\\	&~=~ \term{N_\MX;(\{N/x\}M)_\MX}
\\	&~=~ \term{(N;\{N/x\}M)_\MX}
\\	&~=~ \term{(\{N/x\}(x.M))_\MX}~.
\end{array}
\]

	\item 
Unmarked application case: if $\term{\{N_\MX/x\}M_\MX}=\term{(\{N/x\}M)_\MX}$ and $\term{\{N_\MX/x\}P_\MX}=\term{(\{N/x\}P)_\MX}$ then
\[
\begin{array}{r@{}l}
	\term{\{N_\MX/x\}([P]a.M)_\MX}
	&~=~ \term{\{N_\MX/x\}([P_\MX]a.M_\MX)}
\\	&~=~ \term{[\{N_\MX/x\}P_\MX]a.\{N_\MX/x\}M_\MX}
\\	&~=~ \term{[(\{N/x\}P)_\MX]a.(\{N/x\}M)_\MX}
\\	&~=~ \term{([\{N/x\}P]a.\{N/x\}M)_\MX}~.
\end{array}
\]

	\item 
Unmarked abstraction case: if $\term{\{N_\MX/x\}M_\MX}=\term{(\{N/x\}M)_\MX}$ then
\[
\begin{array}{r@{}l}
	\term{\{N_\MX/x\}(a<y>.M)_\MX}
	&~=~ \term{\{N_\MX/x\}(a<y>.M_\MX)}
\\	&~=~ \term{a<y>.\{N_\MX/x\}M_\MX}
\\	&~=~ \term{a<y>.(\{N/x\}M)_\MX}
\\	&~=~ \term{(a<y>.\{N/x\}M)_\MX}~.
\end{array}		
\]	

	\item 
Marked redex case: if $\term{\{N_\MX/x\}(H.M)_\MX}=\term{(\{N/x\}(H.M))_\MX}$ and  $\term{\{N_\MX/x\}P_\MX}=\term{(\{N/x\}P)_\MX}$ then
\[
\begin{array}{r@{}l}
	&\term{\{N_\MX/x\}([P]a\MXL.H.\MXR a<y>.M)_\MX}
\\	&~=~ \term{\{N_\MX/x\}\{P_\MX/y\}(H.M)_\MX}
\\	&~=~ \term{\{\{N_\MX/x\}P_\MX/y\}\{N_\MX/x\}(H.M)_\MX}
\\	&~=~ \term{\{(\{N/x\}P)_\MX/y\}(\{N/x\}(H.M))_\MX}
\\	&~=~ \term{\{(\{N/x\}P)_\MX/y\}((\{N/x\}H).\{N/x\}M)_\MX}
\\	&~=~ \term{([\{N/x\}P]a\MXL.(\{N/x\}H).\MXR a<y>.\{N/x\}M)_\MX}
\\	&~=~ \term{(\{N/x\}([P]a\MXL.H.\MXR a<y>.M))_\MX}~.
\end{array}		
\]	
\end{itemize}
\end{proof}

\begin{lemma}
\label{lem:marked-reduction-composes}
For a doubly marked term, $\term{M^{\MX\MO}}$, reducing each marking in turn gives the same result as reducing both simultaneously: $\term{(M_\MX)_\MO}=\term{M_{\MX\MO}}$.
\end{lemma}

\begin{proof}
By induction on the size of $\term M$.
\begin{itemize}
	\item 
Unit case: immediate by $\term{((*)_\MX)_\MO}=\term*=\term{(*)_{\MX\MO}}$.

	\item
Variable case: if $\term{(M_\MX)_\MO}=\term{M_{\MX\MO}}$ then
\[
	\term{((x.M)_\MX)_\MO}=\term{x.(M_\MX)_\MO}=\term{x.M_{\MX\MO}}=\term{(x.M)_{\MX\MO}}~.
\]

	\item
Unmarked application case: if $\term{(M_\MX)_\MO}=\term{M_{\MX\MO}}$ and $\term{(N_\MX)_\MO}=\term{N_{\MX\MO}}$ then
\[
\begin{array}{r@{}l}
	\term{(([N]a.M)_\MX)_\MO}
		&~=~\term{[(N_\MX)_\MO]a.(M_\MX)_\MO}
	\\	&~=~\term{[N_{\MX\MO}]a.M_{\MX\MO}}
	\\	&~=~\term{([N]a.M)_{\MX\MO}}~.
\end{array}
\]

	\item
Unmarked abstraction case: if $\term{(M_\MX)_\MO}=\term{M_{\MX\MO}}$ then
\[
	\term{((a<x>.M)_\MX)_\MO}=\term{a<x>.(M_\MX)_\MO}=\term{a<x>.M_{\MX\MO}}=\term{(a<x>.M)_{\MX\MO}}~.
\]

	\item
Doubly-marked redex case: if $\term{(H.M_\MX)_\MO}=\term{H.M_{\MX\MO}}$ and $\term{(N_\MX)_\MO}=\term{N_{\MX\MO}}$ then
\[
\begin{array}{r@{}l}
	\term{([N]a\MXL\MOL.H.\MOR\MXR a<x>.M)_\MX)_\MO}
	&~=~ \term{\{(N_\MX)_\MO/x\}((H.M)_\MX)_\MO}
\\	&~=~ \term{\{N_{\MX\MO}/x\}(H.M)_{\MX\MO}}
\\	&~=~ \term{([N]a\MXL\MOL.H.\MOR\MXR a<x>.M)_{\MX\MO}}~.
\end{array}
\]

	\item
First singly-marked redex case: if $\term{(H.M_\MX)_\MO}=\term{(H.M)_{\MX\MO}}$ and $\term{(N_\MX)_\MO}=\term{N_{\MX\MO}}$ then, using Lemma~\ref{lem:parallel-substitution},
\[
\begin{array}{r@{}l}
	\term{(([N]a\MXL.H.\MXR a<x>.M)_\MX)_\MO}
	&~=~ \term{(\{N_\MX/x\}(H.M)_\MX)_\MO}
\\	&~=~ \term{\{(N_\MX)_\MO/x\}((H.M)_\MX)_\MO}
\\	&~=~ \term{\{N_{\MX\MO}/x\}(H.M)_{\MX\MO}}
\\	&~=~ \term{([N]a\MXL.H.\MXR a<x>.M)_{\MX\MO}}~.
\end{array}
\]
	\item
Second singly-marked redex case: the redex considered is
\[
	\term{[N^{\MX\MO}]a\MOL.H^{\MXL\MOL}.\MOR a<x>.M^{\MXR\MOR}}
\]
where $\term{H}$ does not use the location $\term a$ or a free variable $\term x$. First, it is established that $\term{(H^\MOL.\MOR a<x>.M^\MOR)_\MX}$ is of the form $\term{K^\MOL.\MOR a<x>.N^\MOR}$ where $\term{(H^\MOL.M^\MOR)_\MX}=\term{K^\MOL.N^\MOR}$ for some head context $\term K$ and term $\term N$. For readability the marking $\term{\MO}$ will be suppressed, except on the abstraction $\term{\MOR a<x>}$. The proof is by induction on the size of $\term{H}$.
\begin{itemize}
	\item
Unit case: if $\term H=\term{\{\}}$ then let $\term{K}=\term{\{\}}$ and $\term{N}=\term{M_\MX}$, which gives the following, as required.
\[
\begin{array}{r@{}l@{\qquad}r@{}l}
	\term{(H.\MOR a<x>.M)_\MX} & ~=~ \term{(\MOR a<x>.M)_\MX}
  & \term{(H.M)_\MX}           & ~=~ \term{M_\MX}
\\ &~=~ \term{\MOR a<x>.M_\MX}  && ~=~ \term{N}
\\ &~=~ \term{K.\MOR a<x>.N}	&& ~=~ \term{K.N}~.
\end{array}
\]

	\item
Unmarked application case: if $\term{H}=\term{[P]b.H'}$ then the inductive hypothesis for $\term{H'}$ and $\term{M}$ gives $\term{K'}$ and $\term N$. Let $\term K=\term{[P]b.H'}$, which gives the following, as required.
\[
\begin{array}{l@{\qquad}l}
	\term{([P]b.H'.\MOR a<x>.M)_\MX}		& \term{([P]b.H'.M)_\MX}
\\ ~=~ \term{[P_\MX]b.(H'.\MOR a<x>.M)_\MX}	& ~=~ \term{[P_\MX]b.(H'.M)_\MX}
\\ ~=~ \term{[P_\MX]b.K'.\MOR a<x>.N}		& ~=~ \term{[P_\MX]b.K'.N}
\\ ~=~ \term{K.\MOR a<x>.N}					& ~=~ \term{K.N}~.
\end{array}
\]
	\item
Unmarked abstraction case: if $\term{H}=\term{b<y>.H'}$ then the inductive hypothesis for $\term{H'}$ and $\term{M}$ gives $\term{K'}$ and $\term N$. Let $\term K=\term{b<y>.H'}$, which gives the following, as required.
\[
\begin{array}{l@{\qquad}l}
	\term{(b<y>.H'.\MOR a<x>.M)_\MX}		& \term{(b<y>.H'.M)_\MX}
\\ ~=~ \term{b<y>.(H'.\MOR a<x>.M)_\MX}		& ~=~ \term{b<y>.(H'.M)_\MX}
\\ ~=~ \term{b<y>.K'.\MOR a<x>.N}			& ~=~ \term{b<y>.K'.N}
\\ ~=~ \term{K.\MOR a<x>.N}					& ~=~ \term{K.N}~.
\end{array}
\]

	\item
First marked redex case: if $\term{H}=\term{[P]b\MXL.H'.\MXR b<y>.`{H''}}$ then the inductive hypothesis for $\term{H'.`{H''}}$ and $\term M$ gives $\term{K'}$ and $\term{N'}$. Let $\term{K}=\term{\{P_\MX/x\}K'}$ and $\term{N}=\term{\{P_\MX/x\}N'}$, which gives the following, as required.
\[
\begin{array}{l}
	\term{([P]b\MXL.H'.\MXR b<y>.`{H''}.\MOR a<x>.M)_\MX}
\\  ~=~ \term{\{P_\MX/y\}(H'.`{H''}.\MOR a<x>.M)_\MX}    
\\  ~=~ \term{\{P_\MX/y\}(K'.\MOR a<x>.N')}				 
\\  ~=~ \term{(\{P_\MX/y\}K').\MOR a<x>.\{P_\MX/y\}N'}   
\\  ~=~ \term{K.\MOR a<x>.N}							 
\end{array}
\qquad
\begin{array}{l}
   \term{([P]b\MXL.H'.\MXR b<y>.`{H''}.M)_\MX}
\\ ~=~ \term{\{P_\MX/y\}(H'.`{H''}.M)_\MX}
\\ ~=~ \term{\{P_\MX/y\}(K'.N')}
\\ ~=~ \term{(\{P_\MX/y\}K').\{P_\MX/y\}N'}
\\ ~=~ \term{K.N}~.
\end{array}
\]

	\item
Second marked redex case: if $\term{H}=\term{[P]b\MXL.H'}$ and $\term{M}=\term{`{H''}.\MXR b<y>.M'}$ then the inductive hypothesis for $\term{H'}$ and $\term{`{H''}.M}$ gives $\term{K'}$ and $\term{N'}$. Let $\term{K}=\term{\{P_\MX/x\}K'}$ and $\term{N}=\term{\{P_\MX/x\}N'}$, which gives the following, as required.
\[
\begin{array}{l}
	\term{([P]b\MXL.H'.\MOR a<x>.`{H''}.\MXR b<y>.M)_\MX}
\\  ~=~ \term{\{P_\MX/x\}(H'.\MOR a<x>.`{H''}.M)_\MX}	 
\\  ~=~ \term{\{P_\MX/x\}(K'.\MOR a<x>.N')}              
\\  ~=~ \term{(\{P_\MX/x\}K').\MOR a<x>.\{P_\MX/x\}N')}  
\\  ~=~ \term{K.\MOR a<x>.N}                             
\end{array}
\qquad
\begin{array}{l}
	\term{([P]b\MXL.H'.`{H''}.\MXR b<y>.M)_\MX}
\\  ~=~ \term{\{P_\MX/x\}(H'.`{H''}.M)_\MX}
\\  ~=~ \term{\{P_\MX/x\}(K'.N')}
\\  ~=~ \term{(\{P_\MX/x\}K').\{P_\MX/x\}N')}
\\  ~=~ \term{K.N}~.
\end{array}
\]
\end{itemize}
This concludes the subproof. Returning to the main proof, let $\term{(H.\MOR a<x>.M)_\MX}=\term{K.\MOR a<x>.N}$ where $\term{(H.M)_\MX}=\term{K.N}$. The case concludes as follows.
\[
\begin{array}{r@{}ll}
	\term{(([N]a\MOL.H.\MOR a<x>.M)_\MX)_\MO}
	&~=~ \term{([N_\MX]a\MOL.(H.\MOR a<x>.M)_\MX)_\MO}
\\	&~=~ \term{([N_\MX]a\MOL.K.\MOR a<x>.N)_\MO}
\\	&~=~ \term{\{(N_\MX)_\MO/x\}(K.N)_\MO}
\\	&~=~ \term{\{N_{\MX\MO}/x\}(H.M)_{\MX\MO}}
\\	&~=~ \term{([N]a\MOL.H.\MOR a<x>.M)_{\MX\MO}}
\end{array}
\]
\end{itemize}
\end{proof}

\begin{lemma}
\label{lem:parallel diamond}
Parallel reduction is diamond: if 
$
	\term N\rwlp\term M \rwp\term P
$
then
$
	\term N\rwp\term Q\rwlp\term P
$
for some $\term Q$.
\end{lemma}

\begin{proof}
Let $\term{M^\MX}\rwp\term{M_\MX}=\term N$ and $\term{M^\MO}\rwp\term{M_\MO}=\term P$ for two separate markings $\term{\MX}$ and $\term{\MO}$ of $\term{M}$. Then with both markings on $\term M$ the peak becomes
\[
	\term{(M_\MX)^\MO}\rwlp\term{M^{\MX\MO}}\rwp\term{(M_\MO)^\MX}~.
\]
Let $\term{Q}=\term{M_{\MX\MO}}$ so that by Lemma~\ref{lem:marked-reduction-composes} the peak converges as
\[
	\term{(M_\MX)^\MO}\rwp\term{M_{\MX\MO}}\rwlp\term{(M_\MO)^\MX}~.
\]
\end{proof}

\noindent\textbf{Theorem~\ref{thm:confluence} (restatement)}\quad
Reduction $\rw$ is confluent.

\begin{proof}
A reduction step $\term M\rw\term N$ is a parallel step $\term{M^\MX}\rwp\term{M_\MX}=\term N$ by marking only the reduced redex. A peak $\term M\rwls\term N\rws\term P$ in regular reduction is then immediately one in parallel reduction, $\term{M}\rwlps\term N\rwps\term P$. By the diamond property, Lemma~\ref{lem:parallel diamond}, this converges with parallel reductions, $\term M\rwps\term Q\rwlps\term P$. By Proposition~\ref{prop:parallel correct} a parallel reduction $(\rwps)$ is a regular reduction $(\rws)$, so that the peak converges as $\term M\rws\term Q\rwls\term P$.
\end{proof}

\end{document}
